\documentclass[11pt]{article}
\usepackage{amsmath,amssymb,amsthm}
\usepackage{graphicx}
\usepackage{enumerate}
\usepackage{mathtools}
\usepackage[numbers]{natbib}
\usepackage{mathtools}
\usepackage{comment}
\usepackage{paralist}
\usepackage{nicefrac}
\usepackage{todonotes}

\usepackage{subcaption}
\usepackage[noend,boxed,vlined,ruled,linesnumbered]{algorithm2e} 

\usepackage[colorlinks,citecolor=green!60!black,linkcolor=red!60!black]{hyperref}

\newcommand{\pref}{\succ}
\newcommand{\score}{{{\mathrm{score}}}}
\newcommand{\naturals}{{{\mathbb{N}}}}

\newcommand{\skippedproof}[1]{}

\newcommand{\np}{{\mathrm{NP}}}

\newtheorem{theorem}{Theorem}[section]

\newtheorem{corollary}[theorem]{Corollary}
\newtheorem{lemma}[theorem]{Lemma}
\newtheorem{remark}{Remark}
\newtheorem{observation}{Observation}
\newtheorem{definition}{Definition}[section]

\newtheorem{example}{Example}[section]
\newtheorem{proposition}[theorem]{Proposition}

\newcommand{\calR}{\mathcal{R}}

\newcommand{\calS}{\mathcal{S}}

\newcommand{\cc}{{{{\mathrm{CC}}}}}
\newcommand{\topkcc}{{{{\alpha_k\hbox{-}\mathrm{CC}}}}}
\newcommand{\bordacc}{{{{\beta\hbox{-}\mathrm{CC}}}}}

\newcommand{\topkpav}{{{{\alpha_k\hbox{-}\mathrm{PAV}}}}}
\newcommand{\toptpav}{{{{\alpha_t\hbox{-}\mathrm{PAV}}}}}
\newcommand{\qbordapav}{{{{q\hbox{-}\mathrm{HB}}}}}

\newcommand{\ellpborda}{{{{\ell_p\hbox{-}\mathrm{Borda}}}}}

\newcommand{\sntv}{{{{\mathrm{SNTV}}}}}
\newcommand{\bloc}{{{{\mathrm{Bloc}}}}}
\newcommand{\kborda}{{{{k\hbox{-}\mathrm{Borda}}}}}
\newcommand{\perf}{{{\mathrm{Perf}}}}

\newcommand{\pos}{{{{\mathrm{pos}}}}}

\def\R{\mathbb{R}}

\def\N{\mathbb{N}}

\def\row#1#2{{#1}_1,\ldots ,{#1}_{#2}}

\def\2vec#1#2{\left(\begin{array}{c}{#1}\\{#2}\end{array}\right)}

\newcommand{\reals}{\mathbb{R}}

\title{ %
Committee Scoring Rules: Axiomatic Characterization and  Hierarchy
  \thanks{A preliminary version of this paper appeared in the \emph{Proceedings of the 25th International Joint Conference on Artificial Intelligence (IJCAI-16)~\cite{preliminary}}.}}

\author{
  Piotr Faliszewski\\ 
  AGH University\\ 
  Krakow, Poland \\
  \texttt{faliszew@agh.edu.pl}
\and 
  Piotr Skowron\\
  Technische Universit\"at Berlin\\
  Berlin, Germany\\
  \texttt{p.k.skowron@gmail.com}
\and
  Arkadii Slinko\\
  University of Auckland\\
  Auckland, New Zealand\\
  \texttt{a.slinko@auckland.ac.nz}
\and 
  Nimrod Talmon\\
  Weizmann Institute of Science\\
  Rehovot, Israel\\
  \texttt{nimrodtalmon77@gmail.com}
}

\begin{document}

\maketitle

\begin{abstract}
  Committee scoring voting rules are multiwinner analogues of positional
  scoring rules which constitute an important subclass of
  single-winner voting rules.  We identify several natural subclasses
  of committee scoring rules, namely, weakly separable,
  representation-focused, top-$k$-counting, OWA-based, and
  decomposable rules. We characterize SNTV, Bloc, and $k$-Approval
  Chamberlin--Courant as the only nontrivial rules in pairwise
  intersections of these classes.  We provide some axiomatic
  characterizations for these classes, where monotonicity properties
  appear to be especially useful.  The class of decomposable rules is
  new to the literature. We show that it strictly contains the class
  of OWA-based rules and describe some of the applications of decomposable rules.
\end{abstract}

\section{Introduction}

Axiomatic studies of multiwinner voting rules
go back to Felsenthal and Maoz~\cite{fel-mao:j:norms} and
Debord~\cite{deb:j:k-borda}, but a systematic work on the topic has began only
recently and on several different fronts.  New results appear within social choice theory,
computer science, artificial intelligence, and a number of other
fields (see the work of Faliszewski et al.~\cite{fal-sko-sli-tal:b:multiwinner-trends} for more
details on the history as well as recent progress).
The reason for this explosion of interest from a number of research
communities is the wide range of applications of
multiwinner %
voting rules on the one hand, and the corresponding richness and
diversity of the spectrum of those rules on the other.  Typically,
social-choice theorists study normative properties of various
multiwinner rules, computer scientists investigate feasibility of computing the
election results, and researchers working within artificial
intelligence use multiwinner elections as a versatile tool (e.g.,
useful in genetic
algorithms~\cite{fal-saw-sch-smo:j:multiwinner-genetic}, for ranking
search results~\cite{sko-lac-bri-pet-elk:c:proportional-rankings}, or
for providing personalized
recommendations~\cite{bou-lu:c:chamberlin-courant}).  Yet, there is a
growing interplay between these areas and an increased need for a new
level of comprehension of results obtained in all of
them. %
In this paper we partially address this need by linking syntactic
features of certain families of committee scoring rules with their
normative properties.  The syntactic features of the rules are useful,
e.g., for establishing their computational
properties~\cite{sko-fal-lan:c:collective,fal-sko-sli-tal:c:top-k-counting},
or for viewing those rules as achieving certain optimization goals (which
allows one to consider these rules as tools for certain tasks from
artificial intelligence and operation research).
The normative properties, on the other hand, are useful for
understanding the `behavior' of these rules and the settings for which
they may be appropriate.

The model of multiwinner elections studied in this paper is as
follows. We are given a set of candidates, a collection of
voters---each with a preference order in which the candidates are ranked from the best
to the worst---and an integer~$k$. %
A multiwinner rule 
maps this input to a subset of $k$ candidates (i.e., a committee; we discuss tie-breaking later)
that, in some sense, best reflects the voters' preferences. 
For example, the Single Non-Transferable Vote rule (the SNTV rule)
chooses $k$ candidates that are top-ranked most frequently, whereas
the Bloc rule selects $k$ candidates that are ranked most frequently among top~$k$
positions  (equivalently, under Bloc each voter names
members of his or her favorite committee, and those that are mentioned
most often are selected).
Naturally, there are many other multiwinner rules to choose from, 
defined in various ways.

In this paper we
focus on the class of committee scoring rules, 
introduced by Elkind et
al.~\cite{elk-fal-sko-sli:j:multiwinner-properties} as multiwinner
generalizations of classic positional scoring rules. The main idea of
committee scoring rules is essentially the same as in the
single-winner case: Each voter gives each committee a score based on
the positions of members of this committee in the voter's ranking, 
scores from individual voters are aggregated into the societal scores
of the committees, and the committee(s) with the highest score wins.
Committee scoring rules appear to form a remarkably rich class that
includes both very simple rules, such as SNTV and Bloc, and rather
sophisticated ones, such as the rule of Chamberlin and
Courant~\cite{cha-cou:j:cc} or variants of the Proportional Approval
Voting rule~\cite{kil-handbook}. As these rules tend to be very different in nature,
they are suitable for different purposes,
such as selecting a diverse committee, selecting a committee that
proportionally represents the electorate, or selecting a committee
consisting of $k$ individually best candidates.
This richness is the main strength of the class of committee scoring
rules, but to choose rules for given settings wisely, it is important
to understand the internal structure of the class.
Understanding this structure  is the main goal of the current paper.

So far, researchers have identified the following subclasses of
committee scoring rules (we provide their formal definitions in
Sections~\ref{sec:prelim} and~\ref{sec:classification}; here we give
intuitions only).  (Weakly) separable rules, introduced by Elkind et
al.~\cite{elk-fal-sko-sli:j:multiwinner-properties}, are those rules
where we compute a separate score for each candidate (using a
single-winner scoring rule) and then pick $k$ candidates with the top
scores (for example, using Plurality scores leads to
SNTV).\footnote{If the underlying single-winner scoring rule does not
  depend on the size of the committee (as in the case of SNTV) then
  the rule is referred to as separable. If there is such dependence
  (as in the case of Bloc), then the rule is weakly separable.}
Representation-focused rules, also introduced by Elkind et
al.~\cite{elk-fal-sko-sli:j:multiwinner-properties}, are similar in
spirit to the Chamberlin--Courant rule, whose aim is to ensure that in the
elected committee each voter's most preferred committee member (his or
her representative) is ranked as high as possible. On the other
hand, top-$k$-counting rules, introduced by Faliszewski et
al.~\cite{fal-sko-sli-tal:c:top-k-counting}, capture rules where each
voter evaluates the quality of a committee by the number of members of that committee 
that he or she ranks among the top $k$ ones;
Bloc is a
prime example of a top-$k$-counting rule.  
Finally, the class of OWA-based rules---introduced by Skowron et
al.~\cite{sko-fal-lan:c:collective}, also studied in the approval-based
election
model~\cite{azi-bri-con-elk-fre-wal:j:justified-representation,azi-gas-gud-mac-mat-wal:c:approval-multiwinner,lac-sko:t:approval-thiele}---contains
all the previously mentioned classes. Under these rules a voter
calculates the score of a committee as
the \emph{ordered weighted average} (OWA) of the scores of
the candidates in that committee.\footnote{See the
  original work of Yager~\cite{yag:j:owa} for a general discussion of
  OWAs, and, e.g., the works of Kacprzyk et
  al.~\cite{kac-nur-zad:b:owa-social-choice} or Goldsmith et
  al.~\cite{gol-lan-mat-per:c:rank-dependent} for their other
  applications in voting.}
In this paper we also introduce the class of \emph{decomposable}
committee scoring rules that strictly contains all the OWA-based ones,
has interesting applications, and appears to be easier to work with axiomatically.

All these classes have been defined purely in terms of the syntactic
features of the functions used to calculate the scores of the
committees.\footnote{%
  A notable exception is the class of top-$k$-counting
  rules which were discovered while characterizing those committee scoring
  rules that satisfy the fixed-majority
  property~\cite{fal-sko-sli-tal:c:top-k-counting}.}  These syntactic
features are important if, for example, one wants to assess
some computational properties of the rules (e.g., it is known that weakly
separable rules are polynomial-time
computable~\cite{elk-fal-sko-sli:j:multiwinner-properties}, that
representation-focused rules tend to be $\np$-hard to
compute~\cite{pro-ros-zoh:j:proportional-representation,bou-lu:c:chamberlin-courant,sko-fal-lan:c:collective},
and that the structure of the functions used within OWA-based rules
affects the ability to compute their results
approximately~\cite{sko-fal-lan:c:collective}). Such syntactic features are also
essential when we view committee scoring rules as specifying
optimization goals for particular applications
(for example, since under the Chamberlin--Courant rule each voter's
score depends solely on his or her representative in the elected
committee, this rule is particularly suitable in the context of
deliberative democracy~\cite{cha-cou:j:cc}, for targeted
advertising~\cite{bou-lu:c:chamberlin-courant,bou-lu:c:value-directed-cc},
or for certain facility location
problems~\cite{far-hek:b:facility-location}).  Nonetheless, these
syntactic features do not tell us much about the behavior of
the rules.

Our first result reinforces the syntactic hierarchy of committee
scoring rules. We show that the class of committee scoring rules strictly contains
the class of decomposable rules, which, in turn, strictly contains the class of
OWA-based ones, and that the class of OWA-based rules strictly
contains the classes of (weakly) separable rules,
representation-focused rules, and top-$k$-counting rules. For each
pair of the latter three classes, we show that their intersection
contains exactly one, previously-known, non-trivial voting rule.  See
Figure~\ref{fig:hierarchy} for a visualization of the syntactic
hierarchy of committee scoring rules.

Our second, and the main, result establishes a link between several levels
of the syntactic hierarchy and respective normative properties. In
other words, we establish axiomatic characterizations of some of the
studied subclasses of committee scoring rules.
Until now, the only result of this form,
which is due to Faliszewski et
al.~\cite{fal-sko-sli-tal:c:top-k-counting}, was a characterization of
fixed-majority consistent committee scoring rules as those
top-$k$-counting rules whose scoring functions satisfy (a relaxed
variant of) the convexity property.
Here, our main result is that many of the syntactic properties of our
rules nicely correspond to certain types of
monotonicity. Specifically, we focus on the \emph{committee
  enlargement monotonicity}\footnote{This notion is also known as
  \emph{committee
    monotonicity}~\cite{elk-fal-sko-sli:j:multiwinner-properties} and
  \emph{enlargement
    consistency}~\cite{bar-coe:j:non-controversial-k-names}. We chose
  a name that is more informative than the former, but which is not
  tied to the realm of resolute rules, as the latter.  In the
  literature on apportionment rules, a related property is often called
  \emph{house monotonicity}~\cite{Puke14a,
    bal-you:b:polsci:representation}.} property, which requires that
if we increase the size of the committee sought in the election, then the new
winning committee should be a superset of the old winning committee,
and on variants of the \emph{non-crossing monotonicity} property,
which requires that if we shift forward some members of a
winning committee within any vote in a way that does not affect the
positions of the remaining members of this committee, then this
committee should still win.
We show that committee enlargement monotonicity characterizes exactly
the class of separable rules among committee scoring rules, and that
non-crossing monotonicity characterizes the class of weakly separable
ones.  Then we introduce top-member monotonicity (a variant of
non-crossing monotonicity restricted within each vote to shifting only
the highest-ranked member of the winning committee) and show that
together with narrow-top consistency (which requires that if there are
at most $k$ candidates that are ever ranked in the top position within
a vote, then these candidates should belong to the winning committee)
it characterizes the class of representation-focused rules.  Finally,
we show that if a committee rule is prefix-monotone (i.e., satisfies a
yet another restricted variant of non-crossing monotonicity) then it
must be decomposable.

The paper is organized as follows. In Section~\ref{sec:prelim} we
describe the model of multiwinner elections, define the class of
committee scoring rules, provide their basic properties, and show
several examples of committee scoring rules.
Section~\ref{sec:classification} is devoted to structural properties
of classes of committee scoring rules. Here we build the
hierarchy of the classes and show results regarding
containments and intersections among these classes.
In Section~\ref{sec:axiomatic_properties} we switch to axiomatic
properties of the rules in the classes of the hierarchy and give
several axiomatic characterizations of those classes. 
Finally, we discuss related work in Section~\ref{sec:related} and
conclude in Section~\ref{sec:concl}.
\section{Multiwinner Elections and Committee Scoring Rules}
\label{sec:prelim}
In this section we set the stage for the discussions provided
throughout the rest of the paper by providing preliminary definitions
as well as introducing the class of committee scoring rules.  For each
positive integer $t$, we write $[t]$ to denote $\{1, \ldots, t\}$. By
$\reals_{+}$ we mean the set of nonnegative real numbers.

\subsection{Preliminaries}

An election is a pair $E = (C,V)$, where $C = \{c_1, \ldots, c_m\}$ is
a set of candidates and $V = (v_1, \ldots, v_n)$ is a collection of
voters. Each voter $v_i$ has a \emph{preference order} $\succ_i$,
expressing his or her ranking of the candidates, from the most
desirable one to the least desirable one.  Given a voter $v$ and a
candidate $c$, by $\pos_v(c)$ we mean the position of $c$ in $v$'s
preference order (the top-ranked candidate has position $1$, the next
one has position $2$, and so on).

A \emph{multiwinner voting rule} is a function $\calR$ that given an
election $E = (C,V)$ and a committee size $k$, $1 \leq k \leq |C|$,
returns a family $\calR(E,k)$ of size-$k$ subsets of $C$, i.e., the
set of committees that tie as winners of the election (we use the
nonunique-winner model or, in other words, we assume that multiwinner
rules are irresolute).  We provide a few concrete examples of multiwinner rules in Section~\ref{sec:rules-examples}.
Most of the multiwinner rules that we study are based on single-winner
scoring functions.  A~\emph{single-winner scoring function} for $m$
candidates is a nonincreasing function $\gamma \colon [m] \rightarrow
\reals_{+}$ that assigns a score value to each position in a
preference order.
Given a preference order $\succ_i$ and a candidate $c$, by the
$\gamma$-score of $c$ (given by voter $v_i$) we mean the value
$\gamma(\pos_{v_i}(c))$.  
The two most commonly used %
scoring functions are the Borda scoring function,
\begin{align*}
\beta_m(i) = m-i,
\end{align*}
and the $t$-approval scoring function,
\begin{align*}
\alpha_t(i) = 
  \begin{cases}
    1       & \quad \text{if~} i \leq t \\
    0       & \quad \text{otherwise.}
  \end{cases}
\end{align*}
In particular, $\alpha_1$ is known as the Plurality scoring function.

Committee scoring functions generalize single-winner scoring functions
to the multiwinner setting in a natural way, by assigning scores to
the positions of the whole committees. Formally, given a vote $v$ and
a committee $S$ of size $k$, %
the \emph{committee position} of $S$ in $v$, denoted $\pos_v(S)$, is a
sequence $(i_1, \ldots, i_k)$ that results from sorting the set $\{
\pos_v(c) \mid c \in S\}$ in the increasing order.  We write $[m]_k$
to denote the set of all such length-$k$ increasing sequences of
numbers from $[m]$ (in other words, we write $[m]_k$ to denote the set
of all possible committee positions for the case of $m$ candidates and
committees of size $k$).  Given two committee positions from $[m]_k$,
$I = (i_1, \ldots, i_k)$ and $J = (j_1, \ldots, j_k)$, we say that $I$
\emph{weakly dominates} $J$, $I \succeq J$, if for each $t \in [k]$,
it holds that $i_t \leq j_t$ (we say that $I$ \emph{dominates} $J$,
denoted $I \succ J$, if at least one of these inequalities is
strict\footnote{In previous papers on committee scoring rules, the
  notions of weak dominance and dominance were conflated. We believe
  that giving them clear, separate meanings will help in providing
  more crisp arguments and discussions.}).  Below we define committee
scoring functions formally.

\begin{definition}[Elkind et al.~\cite{elk-fal-sko-sli:j:multiwinner-properties}]
  A committee scoring function for $m$ candidates and a committee size
  $k$ is a function $f_{m,k} \colon [m]_k \rightarrow \reals_{+}$ such
  that for each two sequences $I,J \in [m]_k$, if $I$ weakly dominates $J$
  then $h(I) \geq h(J)$.
\end{definition}

Let $f=(f_{m,k})_{k\le m}$ be a family of committee scoring functions,
where %
each $f_{m,k}$ is a function for $m$ candidates and committees of size
$k$. Given an election $E = (C, V)$ with $m$ candidates and a committee~$S$ of
size $k$, we define the $f_{m,k}$-score of $S$ to be:
\begin{align*}
f_{m,k}\hbox{-}\score_E(S) = \sum_{v_i \in V}f_{m,k}(\pos_{v_i}(S)) \text{.}
\end{align*}
When $f$ is clear from the context, we often speak of the score of a
committee instead of its $f$-score.
Given the above notation, we are ready to define committee scoring rules
formally.

\begin{definition}
  Let $f=(f_{m,k})_{k\le m}$ be a family of committee scoring
  functions (with one function for each $m$ and $k$, $k \leq
  m$). Committee scoring rule $\calR_f$ is a multiwinner voting rule
  that given an election $E = (C,V)$ and committee size $k$, outputs
  all size-$k$ committees with the highest $f_{|C|,k}$-score.
\end{definition}

We say that a committee scoring rule $\calR_f$ is \emph{degenerate} if
there is a number of candidates $m$ and a committee size $k$ such that
$f_{m,k}$ is a constant function. As a consequence, a degenerate rule returns all
size-$k$ committees for every election with $m$ candidates.
The \emph{trivial committee scoring} rule is a degenerate rule that
returns the set of all size-$k$ committees for all elections and all
sizes $k$ (naturally, it is defined by a family of constant
functions).

\subsection{Examples of Committee Scoring Rules}
\label{sec:rules-examples}

Many well-known multiwinner rules are, in fact, committee scoring
rules; below we provide several such examples. For each of the
rules we provide the family of committee scoring functions used in its
definition, discuss these functions intuitively, and mention some
applications.

\begin{description}
\item[SNTV, Bloc, and $\boldsymbol{k}$-Borda.] These three rules use the following committee
  scoring functions:
  \begin{align*}
  f^\sntv_{m,k}(\row ik) &= \textstyle\sum_{t=1}^{k}\alpha_1(i_t) = \alpha_1(i_1) \text{,} \\
  f^\bloc_{m,k}(\row ik) &= \textstyle\sum_{t=1}^{k}\alpha_k(i_t) \text{, and} \\
  f^\kborda_{m,k}(\row ik) &= \textstyle\sum_{t=1}^k\beta_m(i_t) \text{.}
  \end{align*}
  That is, under the SNTV rule we choose $k$ candidates with the
  highest Plurality scores, under Bloc we choose $k$ candidates with
  the highest $k$-Approval scores, and under $k$-Borda we choose~$k$
  candidates with the highest Borda scores.
  On the intuitive level, under SNTV each voter names his or her
  favorite committee member, 
  under Bloc each voter names all the $k$ members of his or her
  favorite committee, and under $k$-Borda each voter ranks all the
  candidates and assigns them scores in a way which corresponds
  linearly to their position in the ranking. SNTV and Bloc are
  sometimes used in political elections (with the former used, e.g.,
  in the parliamentary elections in Puerto Rico, and with the latter
  often used for various local elections in many countries). $k$-Borda
  and other rules based on similar scoring schemes are often used to
  determine finalists of competitions (e.g., the finalists of the
  Eurovision Song Contest are selected using a system very close to
  $k$-Borda).

\item[The Chamberlin--Courant rule.] Under the Chamberlin--Courant rule
  (the $\beta$-CC rule), the score that a voter $v$ assigns to a
  committee $S$ depends only on how $v$ ranks his or her favorite
  member of~$S$ (referred to as $v$'s \emph{representative} in~$S$). The
  Chamberlin--Courant rule seeks committees in which each voter ranks
  his or her representative as high as possible. Formally, the rule
  uses functions:
  \begin{align*}
    f^\bordacc_{m,k}(\row ik) = \beta_m(i_1) \text{.}
  \end{align*}
  This is the variant of the rule originally proposed by Chamberlin
  and Courant~\cite{cha-cou:j:cc}, but, subsequently, other authors
  (e.g., Procaccia et
  al.~\cite{pro-ros-zoh:j:proportional-representation}, Betzler at
  al.~\cite{bet-sli-uhl:j:mon-cc}, and Faliszewski et
  al.~\cite{fal-sli-sta-tal:c:cc-mon-clustering}) considered other
  ones, based on other single-winner scoring functions. In particular,
  we will be interested in the $k$-Approval Chamberlin--Courant rule,
  ($\alpha_k$-CC) which is defined through functions:
  \begin{align*}
  f^\topkcc_{m,k}(\row ik) = \alpha_k(i_1) \text{.}
  \end{align*}
  Intuitively, both variants of the Chamberlin--Courant rule seek
  committees of diverse candidates that ``cover'' as broad a spectrum
  of voters' views as
  possible. %
  Lu and Boutilier~\cite{bou-lu:c:chamberlin-courant} considered the
  rule in the context of recommendation systems.

\item[The PAV rule.] The Proportional Approval Voting rule (the PAV
  rule) was originally defined by Thiele~\cite{Thie95a} in the
  approval setting (where instead of ranking the candidates, the
  voters indicate which ones they accept as committee members; for
  recent discussions of the rule see the overview of
  Kilgour~\cite{kil-handbook} and the works of Aziz et
  al.~\cite{azi-bri-con-elk-fre-wal:j:justified-representation} and
  Lackner and Skowron~\cite{lac-sko:t:approval-thiele}). We model it
  as a committee scoring rule $\alpha_t$-PAV, where $t$ is a
  parameter, defined using scoring functions of the form:
  \begin{align*}
    f^\toptpav_{m,k}(\row ik) = \alpha_t(i_1) +
    \textstyle\frac{1}{2}\alpha_t(i_2) +
    \frac{1}{3}\alpha_t(i_3) +
    \cdots + 
    \frac{1}{k}\alpha_t(i_k)
    \text{.}
  \end{align*}
  PAV is particularly well-suited for electing parliaments. Indeed,
  Brill et al.~\cite{bri-las-sko:c:apportionment} have shown that it
  generalizes the d'Hondt apportionment method, which is used for this
  purpose in many countries (e.g., in France and Poland).
  A number of recent works~\cite{azi-bri-con-elk-fre-wal:j:justified-representation,bri-las-sko:c:apportionment,lac-sko:t:approval-thiele,azi-elk-hua-lac-san-sko:c:ejr-poly}
  explain why the harmonic sequence used within the PAV scoring
  function ensures that the elected committee represents the voters
  proportionally.
\end{description}

Naturally, there are many other committee scoring rules, and we will
discuss some of them throughout the paper. Nonetheless, the above few
suffice to illustrate our main points.  There is also a number of
other multiwinner rules that are not committee scoring rules, such as
STV (see, e.g., the work of Tideman and
Richardson~\cite{tid-ric:j:stv}), Monroe~\cite{mon:j:monroe}, Minimax
Approval Voting~\cite{bra-kil-san:j:minimax}, or rules which are
stable in the sense of Gehrlein~\cite{geh:j:condorcet-committee}. We
do not discuss them in this paper, but we provide some literature
pointers in Section~\ref{sec:related}.

\subsection{Basic Features of Committee Scoring Rules}
\label{sec:basic}

The class of committee scoring rules is very rich and there are only a
few basic properties shared by all the rules in this class. Below we
discuss several such properties that will be useful throughout this
paper.

From our point of view, the most important common feature of committee
scoring rules is that they are uniquely defined by their scoring
functions (up to linear transformations). Formally, we have the
following lemma (we provide the proof in the appendix).

\newcommand{\lemuni}{Let $\calR_f$ and $\calR_g$ be two committee
  scoring rules defined by committee scoring functions $f =
  (f_{m,k})_{k \leq m}$ and $g = (g_{m,k})_{k \leq m}$,
  respectively. If $\calR_f = \calR_g$ then for each $m$ and~$k$, $k
  \leq m$, there are two values, $a_{m,k} \in \reals_{+}$ and $b_{m,k}
  \in \reals$, such that for each $I \in [m]_k$ we have that
  $f_{m,k}(I) = a_{m,k} \cdot g_{m,k}(I) + b_{m,k}$.}

\begin{lemma}\label{lem:unique}
  Let $\calR_f$ and $\calR_g$ be two committee scoring rules defined
  by committee scoring functions $f = (f_{m,k})_{k \leq m}$ and $g =
  (g_{m,k})_{k \leq m}$, respectively. If $\calR_f = \calR_g$ then for
  each $m$ and $k$, $k \leq m$, there are two values, $a_{m,k} \in
  \reals_{+}$ and $b_{m,k} \in \reals$, such that for each $I \in
  [m]_k$ we have that $f_{m,k}(I) = a_{m,k} \cdot g_{m,k}(I) + b_{m,k}$.
\end{lemma}

Due to Lemma~\ref{lem:unique}, to show that two committee scoring
rules are distinct it suffices to show that their scoring functions
are not linearly related. In particular, this will be very useful when
we will be showing that certain rules cannot be represented using
scoring functions of a given form.

The second common feature of committee scoring rules is nonimposition,
which requires that for every committee there is some election where
it wins uniquely.  Formally, we have the following definition.

\begin{definition}\label{def:nonimposition}
  Let $\calR$ be a multiwinner rule. We say that $\calR$ has the
  nonimposition property if for each candidate set $C$ and each subset
  $W$ of $C$, there is an election $E = (C,V)$ such that $\calR(E,
  |W|) = \{W\}$.
\end{definition}

Nonimposition is such a basic property that it is hardly surprising
that all non-degenerate committee scoring rules satisfy it. We prove
the next lemma in the appendix.

\begin{lemma}\label{lem:nonimposition}
  Let $\calR_f$ be a committee scoring rule defined by a family of
  committee scoring functions $f=(f_{m,k})_{k\le m}$.  $\calR_f$
  satisfies the nonimposition property if and only if every committee
  scoring function in $f$ is nontrivial.
\end{lemma}

While at first sight nonimposition and Lemma~\ref{lem:nonimposition}
seem hardly exciting, in fact they are sufficient to illustrate
intriguing differences between single-winner voting rules and their
multiwinner counterparts. For example, one can verify that all
nontrivial single-winner scoring rules satisfy the following extended
variant of the nonimposition property: For every candidate set $C$ and
its subset $S$, there is an election $E = (C,V)$ where exactly the
candidates from $S$ tie as winners. Analogous result does not hold for
committee scoring rules, even for the case of two committees (in which
case it could be dubbed as $2$-nonimposition; the example below is due
to Lackner and Skowron~\cite{lac-sko:t:approval-thiele}).

\begin{example}
  Let us fix some committee size $k$ and a set $C$ containing at least
  $2k$ candidates.  Consider two disjoint committees $W_1$ and
  $W_2$. Let $E$ be an arbitrary election where $W_1$ and $W_2$ are
  tied as winners according to Bloc (such elections exist). We note
  that each candidate in $W_1$ has exactly the same $k$-Approval score
  as each candidate in $W_2$ (otherwise at least one of these
  committees would not be winning).  Consequently, every size-$k$
  committee $W$ such that $W \subseteq W_1 \cup W_2$ is also winning
  in $E$, so $W_1$ and $W_2$ are not the two unique winning
  committees.
\end{example}

The fact that in general $2$-nonimposition does not hold for committee
scoring rules is quite disappointing because many results would be far
easier to prove if we could assume that it is always possible to
construct an election where two arbitrary given committees are the
only winning ones. On the other hand, it is possible to construct
elections where two size-$k$ committees $W_1$ and $W_2$ are the only
winning ones, provided that they share $k-1$ candidates (and, indeed,
this fact is used in the proof of Lemma~\ref{lem:unique}).

There are a few more common properties of committee scoring rules. For
example, they all satisfy the \emph{candidate
  monotonicity} property which requires that if we shift forward a
member of a winning committee then, afterward, this candidate still
belongs to some winning committee (but possibly quite a different one;
see the work of Bredereck et
al.~\cite{bre-fal-kac-nie-sko-tal:c:multiwinner-robustness}). Also,
all committee scoring rules are \emph{consistent} in the sense that if
two elections $E_1$ and $E_2$ (over the same candidate set) have some
common winning committees, then these are exactly the winning
committees in an election obtained by merging the voter collections of
$E_1$ and $E_2$.  The former property is related to our discussions in
Section~\ref{sec:axiomatic_properties} and the latter one is often
useful as a tool when proving various results (and, indeed, it is
crucial in characterizing the class of committee scoring rules
axiomatically~\cite{sko-fal-sli:t:axiomatic-committee}).

\subsection{The T-Shirt Store Example}

In Section~\ref{sec:rules-examples} we have provided a number of
examples of committee scoring rules and we have discussed some of
their applications, focusing mostly on political elections. However,
committee scoring rules have far more varied applications (see, e.g.,
the overview of Faliszewski et
al.~\cite{fal-sko-sli-tal:b:multiwinner-trends}), most of which have
nothing to do with politics. Below we describe a simplified
business-inspired scenario where committee scoring rules may be
useful.  We use this example to guide our way through the different
types of committee scoring rules discussed in this paper.

\begin{example}\label{example:shirts1}
  Consider a T-shirt store that needs to decide which shirts to put on
  offer.
  Let $C$ be the set of T-shirts that the store can order from its
  suppliers ($|C| = m$). Since the store has limited space, it can
  only put $k$ different T-shirts on display, and it wants to pick
  them in a way that would maximize its revenue (i.e., the number of
  T-shirts sold).  We assume that every customer knows all the designs
  (say, from a website) and ranks all T-shirts from the best one to
  the worst one.
  Let us say
  that a customer considers a T-shirt to %
  be ``very good'' if it is among the top $k$ T-shirts (of course,
  this is an arbitrary choice, made for the sake of simplifying the
  example).

  How should the store decide which T-shirts to put on display? This
  depends on how the customers behave. Consider a customer that ranks
  the available T-shirts on positions $i_1 < i_2 < \cdots < i_k$. If
  this is a very picky customer that only buys a T-shirt if it is the
  very best among all possible ones (according to his or her
  opinion) then the number of T-shirts this customer buys is given by
  $f^\sntv_{m,k}(\row ik) = \alpha_1(i_1)$.  However, if
  this customer were to buy one copy of each T-shirt he or she considered as ``very
  good,'' he or she would buy $f^\bloc_{m,k}(\row ik) =
  \textstyle\sum_{t=1}^{k}\alpha_k(i_t)$ T-shirts. 
  It is also possible that a customer would buy only one $T$-shirt,
  provided he or she considered it as ``very good.'' The number of
  T-shirts bought by such a customer would be $f^\topkcc_{m,k}(\row
  ik) = \alpha_k(i_1)$.  Depending on which type of customers the
  store expects to have, it should choose its selection of T-shirts
  either using SNTV, Bloc, or $k$-Approval
  Chamberlin--Courant. (Surely, other types of customers are possible
  as well and we will discuss some of them later. It is also likely
  that the store would face a mixture of different types of customers,
  but this is beyond our study.)
\end{example}

\section{Hierarchy of Committee Scoring Rules}
\label{sec:classification}

In this section we describe the classes of committee scoring rules
that were studied to date, introduce a new class---the class of
decomposable rules---and argue how all these classes relate to each
other, forming a hierarchy. In Figure~\ref{fig:hierarchy} we present
the relations between the classes discussed in this section, with
examples of notable rules.
The classes are defined by setting restrictions on the scoring
functions so, in other words, in this section we are interested in the
syntactic hierarchy of committee scoring rules.  Later, in
Section~\ref{sec:axiomatic_properties}, we will consider semantic
properties.

\begin{figure}
\center
\scalebox{1.0}{
\begin{tikzpicture}
\tikzstyle{node} = [draw=black,fill=none]

\fill [node] (0,1.35) ellipse (7 and 4.5);
\node [above] at (0,5.05) {\textbf{committee scoring rules}};
\node [above] at (0,4.55) {\large{$\max$-threshold rules, $\ell_p$-Borda}};

\fill [node] (0,0.675) ellipse (6.5 and 3.75);
\node [above] at (0,3.75) {\large\textbf{decomposable}};
\node [above] at (0,3.25) {\large multithreshold rules};

\fill [node] (0,0) ellipse (6 and 3);
\node [above] at (0,2.25) {\large\textbf{OWA-based}};
\node [above] at (0,1.75) {\large $\alpha_t$-PAV, $q$-HarmonicBorda};

\fill [node] (0,-1.425) ellipse (3 and 1.5);
\node [above] at (0,-2.1) {\textbf{representation-focused}};
\node [above] at (0,-2.6) {\small{$\beta$-CC}};

\fill [node] (-2.0,0) ellipse (3 and 1.5);
\node [above] at (-2.25,0.5) {\textbf{weakly separable}};
\node [above] at (-2.25,0) {{$k$-Borda}};

\fill [node] (2.0,0) ellipse (3 and 1.5);
\node [above] at (2.25,0.55) {\textbf{top-$\boldsymbol{k}$-counting}};
\node [above] at (2.95,0) {{$\alpha_k$-PAV, Perfectionist}};

\node [above] at (0,0.25) {{Bloc}};

\node [above] at (1.45,-1.2) {{$\alpha_k$-CC}};

\node [above] at (-1.45,-1.2) {{SNTV}};

\node [above] at (0,-0.7) {\small{Trivial}};

\end{tikzpicture}
}
\caption{\label{fig:hierarchy}The hierarchy of committee scoring rules.}
\end{figure}
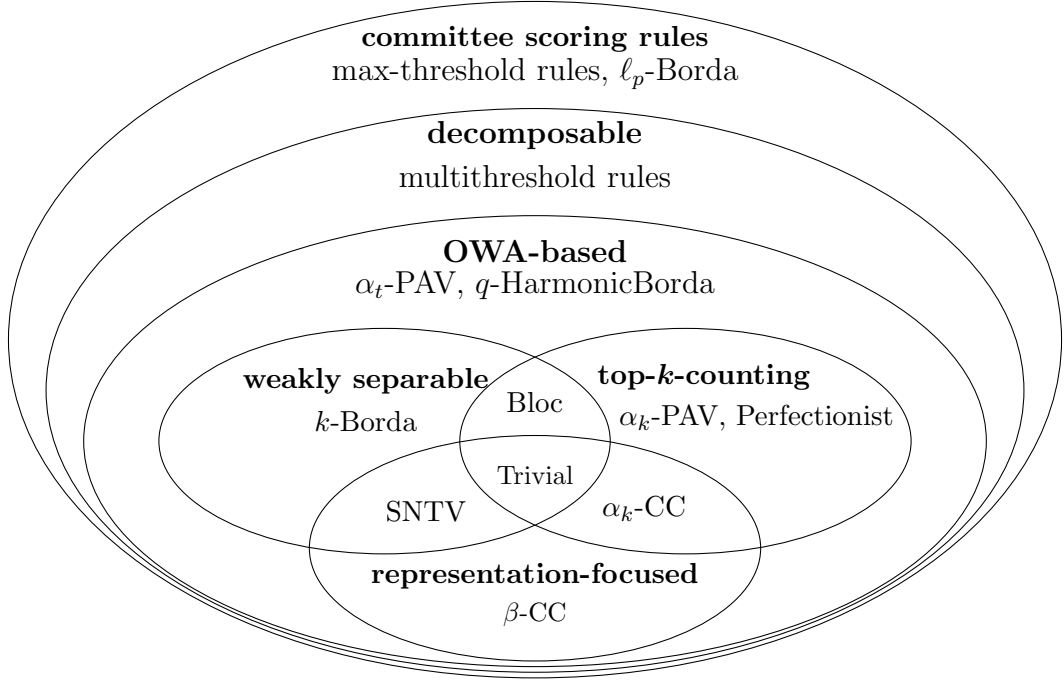

\subsection{Separable and Weakly Separable Rules}

We say that a family of committee scoring functions $f=(f_{m,k})_{k\le
  m}$ is {\em weakly separable} if there exists a family of
(single-winner) scoring functions $(\gamma_{m,k})_{k\le m}$ with
$\gamma_{m,k}\colon [m] \to \R_+$ such that for every $m\in \N$ and
every committee position $I = (i_1, \ldots, i_k) \in [m]_k$ we have:
\[
f_{m,k}(i_1, \ldots, i_k) = \textstyle\sum_{t=1}^k \gamma_{m,k}(i_t).
\]
A committee scoring rule $\calR_f$ is {\em weakly separable} if it is
defined through a family of weakly separable scoring functions $f$.
In other words, if a rule is weakly separable then we can compute the
score of each candidate independently, using the single-winner scoring
function $\gamma_{m,k}$, and pick the $k$ candidates with the highest
scores.
In consequence, it is possible to compute winning committees for all
weakly separable rules in polynomial time, provided that their
underlying single-winner scoring functions are polynomial-time
computable~\cite{elk-fal-sko-sli:j:multiwinner-properties}.\footnote{There
  is a subtlety here as there may be exponentially many winning
  committees. However, by listing the scores of all the candidates, we
  provide enough information to, e.g., enumerate all the winning
  committees in time proportional to the number of these committees,
  or to perform many other tasks related to winner determination (such
  as computing the score of a winning committee).}

If for all $m$ we have $\gamma_{m,1} = \cdots = \gamma_{m,m}$, then we
say that the family $f$ and the corresponding committee scoring rule $\calR_f$ are {\em separable}, without the ``weakly''
qualification. Thus, separable rules use the same scoring function for each value of the size of a committee to be elected.
Interestingly, separable rules have some axiomatic properties that
other weakly separable %
rules lack~\cite{elk-fal-sko-sli:j:multiwinner-properties}---we will discuss this further in Section~\ref{sec:axiomatic_properties}.

The notion of (weakly) separable rules was introduced by Elkind et
al.~\cite{elk-fal-sko-sli:j:multiwinner-properties}; they pointed
out that SNTV and $k$-Borda are separable, whereas Bloc is only weakly
separable.

\subsection{Representation-Focused Rules}

A family of committee scoring functions $f=(f_{m,k})_{k\le m}$ is {\em
  representation-focused} if there exists a family of (single-winner)
scoring functions $(\gamma_{m,k})_{k\le m}$ such that for every $m\in
\N$ and every committee position $I = (i_1, \ldots, i_k) \in [m]_k$ we
have:
\begin{align*}
   f_{m,k}(i_1, \ldots, i_k) = \gamma_{m,k}(i_1).
\end{align*}
This means that the score that a committee receives from a voter
depends only on the position of the most preferred member of this
committee in the voter's preference ranking---such a member can be
viewed as a representative of the voter in the committee.  A committee
scoring rule $\calR_f$ is {\em representation-focused} if it is
defined through a family of representation-focused scoring
functions~$f$.  The notion of representation-focused rules was
introduced by Elkind et
al.~\cite{elk-fal-sko-sli:j:multiwinner-properties}; $\beta$-CC is the
archetypal example of a representation-focused committee scoring rule
and, in consequence, all the representation-focused rules can be seen as
variants of the Chamberlin--Courant rule.

SNTV is both separable and representation-focused, and it is the only
non-degenerate committee scoring rule with this property.

\begin{proposition}
  SNTV is the only non-degenerate committee scoring rule that is
  (weakly) separable and representation-focused.
\end{proposition}
\begin{proof}
  It is easy to verify that SNTV is separable and
  representation-focused.  For the other direction, let $\calR$ be a
  rule which is separable and representation focused.  It follows that
  $\calR \equiv \calR_f \equiv \calR_g$ for some families of committee
  scoring functions $f$ and $g$, such that $f_{m,k}(i_1, \ldots, i_k)
  = \phi_{m,k}(i_1) + \ldots + \phi_{m,k}(i_k)$ and $g(i_1, \ldots,
  i_k) = \gamma_{m,k}(i_1)$. Every linear transformation of $g$ has
  the same form (i.e., it only depends on $i_1$), so by
  Lemma~\ref{lem:unique} (linearly transforming $g$, if necessary)
  we can assume that $f = g$.

  Without loss of generality, we can assume that $m > k$. For each
  committee positions $(i_1, \ldots, i_k)$ with $i_1 = 1$, we have
  that
  \begin{align*}
  \phi_{m,k}(i_1) + \ldots + \phi_{m,k}(i_k) = \gamma_{m,k}(i_1) \text{,}
  \end{align*}
  and, so, we can conclude that $f_1(2) = \cdots = f_1(m)$.  Since
  $\calR$ is non-degenerate, we have that $f_1(1) > f_1(m)$, and so
  that $f_1(1) > f_1(2)$.  This is sufficient to conclude that $\calR$
  is equivalent to SNTV.
\end{proof}

Generally, representation-focused rules are $\np$-hard to compute
(SNTV is one obvious exception).  This fact was first shown by Procaccia et
al.~\cite{pro-ros-zoh:j:proportional-representation} in the
approval-based setting, and then by Lu and
Boutilier~\cite{bou-lu:c:chamberlin-courant} for $\bordacc$.  Since
then, various means of computing the results under the
Chamberlin--Courant rule and its variants were studied in quite some
detail~\cite{bet-sli-uhl:j:mon-cc,cor-gal-spa:c:sp-width,sko-fal-sli:j:multiwinner,sko-yu-fal-elk:j:mwsc,pet:t:total-unimodularity,fal-sli-sta-tal:c:cc-mon-clustering,lac-pet:c:spoc,fal-lac-pet-tal:c:csr-heuristics}.

\subsection{Top-$\boldsymbol{k}$-Counting Rules}

A committee scoring rule $\calR_f$, defined by a family
$f=(f_{m,k})_{k\le m}$, is {\em top-$k$-counting} if there exists
a %
sequence of nondecreasing functions $(g_{m,k})_{k \leq m}$, with $g_{m,k} \colon
\{0, \ldots, k\} \rightarrow \reals_{+}$, such that: %
\[
  f_{m,k}( i_1, \ldots, i_k ) = g_{m,k}\Big( \big| \{ i_t \mid i_t \leq k\} \big| \Big). 
\]
That is, the value $f_{m,k}(i_1, \ldots, i_k)$ depends only on
the number of committee members that the given voter ranks among his
or her top $k$ positions.  We refer to the functions $g_{m,k}$ as the {\em
  counting functions}.  Top-$k$-counting rules were introduced by
Faliszewski et al.~\cite{fal-sko-sli-tal:c:top-k-counting}.

\begin{remark}
  It would be quite natural to require that all counting functions for
  a given committee size were the same, that is, that for each $k \in
  \naturals$ it held that $g_{k,k} = g_{k+1,k} = g_{k+2,k} =
  \cdots$. Following Faliszewski et
  al.~\cite{fal-sko-sli-tal:c:top-k-counting}, we formally do not make
  this requirement, but we expect it to hold for all natural
  top-$k$-counting rules.
\end{remark}

Top-$k$-counting rules include, for example, the Bloc rule,
$\alpha_k$-PAV, and $\alpha_k$-CC, where Bloc uses the linear
counting functions $g_{m,k}^\bloc(i) = i$, 
$\alpha_k$-PAV uses counting functions $g_{m,k}^\topkpav(i) =
\sum_{t=1}^i\frac{1}{t}$, and $\alpha_{k}$-CC uses counting functions:
\begin{align*}
g_{m,k}^\cc(i) =
  \begin{cases}
    1       & \quad \text{if~} i \geq 1\\
    0       & \quad \text{if~} i = 0.
  \end{cases}
\end{align*}
As an extreme example of a top-$k$-counting rule, Faliszewski et
al.~\cite{fal-sko-sli-tal:c:top-k-counting} introduced the
Perfectionist rule, which uses counting functions:
\begin{align*}
g_{m,k}^\perf(i) =
  \begin{cases}
    1       & \quad \text{if~} i = k\\
    0       & \quad \text{otherwise}
  \end{cases}
\end{align*}
Perfectionist is extreme in the sense that a voter assigns a point to
a committee exactly if he or she ranks all the members of this
committee as $k$ best ones.

\begin{example}\label{example:shirts-perfectionist}
  Let us recall our T-shirt store example
  (Example~\ref{example:shirts1}).  Consider a particularly snobbish
  customer, who is willing to buy a shirt from a store only if he or
  she views all the available shirts as very good (recall that we
  defined ``very good'' to mean being ranked among top $k$
  positions). Then if $i_1, \ldots, i_k$ are the positions of the
  available shirts in the customer's ranking, the number of shirts
  that the store should expect to sell to such a customer is:
  \[
  f_{m,k}(i_1, \ldots, i_k) = g_{m,k}^\perf\Big( \big| \{ i_t \mid i_t
  \leq k\} \big| \Big) = \alpha_k(i_k).
  \]
  Thus if the store expects such customers, then it should use the
  Perfectionist rule to choose its merchandise (and, possibly,
  should also increase its prices!).
\end{example}

Bloc is the only nontrivial rule that is both weakly separable and
top-$k$-counting, and $\alpha_k$-CC is the only nontrivial rule that
is both representation-focused and top-$k$-counting.

\begin{proposition} %
  Bloc is the only nontrivial %
  rule that is weakly separable and top-$k$-counting.
\end{proposition}

\begin{proof}
  By combining Lemma~\ref{lem:unique} and the results of Faliszewski
  et al.~\cite{fal-sko-sli-tal:c:top-k-counting}, we obtain that
  top-$k$-counting rule defined by a family of linear counting
  functions is the only weakly separable top-$k$-counting rule, and this
  rule is exactly Bloc.
\end{proof}

\begin{proposition}\label{thm:cc-hope}
  $\alpha_k$-CC is the only nontrivial %
  rule that is representation-focused and top-$k$-counting.
\end{proposition}
\begin{proof}
  It is easy to verify that $\alpha_k$-CC is top-$k$-counting and
  representation-focused.  For the other direction, let $\calR$ be a
  rule which is both top-$k$-counting and representation focused.  It
  follows that $\calR \equiv \calR_f \equiv \calR_g$ for two
  functions, $f$ and $g$, such that, $f(i_1, \ldots, i_k) = f_1(i_1)$ and
  $g(i_1, \ldots, i_k) = g_1(s)$, where $s = \big|\{i_t \mid i_t \leq k\}\big|$.
  Since any linear transformation of $f$ has the same form, by the uniqueness %
  we can assume that $f = g$.

  For each $i \in [m-k+1]$ let $L(i)$ denote the sequence $(i, m-k+2, m-k+3, \ldots, m)$. For each $i, j > k$ we have that:
  \begin{align*}
  f_1(i) = f(L(i))= g(L(i)) = g(L(j)) = f(L(j)) = f_1(j) \text{.}
  \end{align*}
  By the same reasoning, we can prove that for each $i, j \leq k$ we have $f_1(i) = f_1(j)$.
  Since the rule is nontrivial, we know that for some $i, j$ it holds that $f_1(i) \neq f_1(j)$.
  This is sufficient to claim that $\calR$ is equivalent to $\alpha_k$-CC.
\end{proof}

Faliszewski et al.~\cite{fal-sko-sli-tal:c:top-k-counting} show that
top-$k$-counting rules tend to be $\np$-hard to compute, but point out
several polynomial-time computable exceptions, including Bloc and
Perfectionist. They also observe that for rules with concave counting
functions there are polynomial-time constant-factor approximation
algorithms, whereas for rules with convex counting functions such
algorithms may be missing (under standard complexity-theoretic
assumptions).

\subsection{OWA-Based Rules}
Skowron et al.~\cite{sko-fal-lan:c:collective} introduced a class of
multiwinner rules based on ordered weighted average (OWA) operators.
Similar rules for approval-based ballots were first considered in the
19th century by Thiele~\cite{Thie95a} and more recently were studied
by Aziz et
al.~\cite{azi-bri-con-elk-fre-wal:j:justified-representation,azi-gas-gud-mac-mat-wal:c:approval-multiwinner}
and Lackner and Skowron~\cite{lac-sko:t:approval-thiele} (see also the
discussion by Kilgour~\cite{kil-handbook}).  Elkind and
Ismaili~\cite{conf/aldt/ElkindI15} use OWA operators to define a
different class of multiwinner rules, which we do not consider in this
paper.

We provide intuition for the OWA-based rules by using our T-shirts
store example.

\begin{example}\label{example:shirts}
  Let us say that a customer views a T-shirt as ``good enough'' if it
  is among the top $10\%$ of the shirts available on the market.  Suppose
  that a customer identifies the best T-shirt available in the store
  and buys it with probability 1, provided it is ``good enough''. Then he
  or she also finds the second best T-shirt and buys it with
  probability $\nicefrac{1}{2}$ (again, provided that it is ``good
  enough''), the third best shirt with probability $\nicefrac{1}{3}$,
  and so on, all the way to the $k$'th best T-shirt, which he or she buys
  with probability $\nicefrac{1}{k}$ (if it is ``good enough'').  If $i_1,
  \ldots, i_k$ are the positions (in the customer's preference order)
  of the T-shirts that the store puts on display, then the expected
  number of T-shirts he or she buys is given by the function:
   \[
     f_{m,k}(i_1, \ldots, i_k) = 1\cdot \alpha_{0.10m}(i_1) + \nicefrac{1}{2}\cdot \alpha_{0.10m}(i_2)+\cdots+\nicefrac{1}{k}\cdot \alpha_{0.10m}(i_k) .
   \]
   Thus, to maximize its revenue, the store should find a winning
   committee for the election where the T-shirts are the candidates,
   the voters are the customers, and where we use committee scoring
   rule $\calR_f$ based on $f = (f_{m,k})_{k \le m}$. This multiwinner
   voting rule is $\alpha_{0.10m}$-PAV, a variant of the Proportional
   Approval Voting rule.
\end{example}

Now let us define OWA-based rules formally. An OWA operator $\Lambda$
of dimension $k$ is a sequence $\Lambda = (\lambda_1, \ldots,
\lambda_k)$ of nonnegative real numbers.

\begin{definition}\label{def:owa-csr}
  Let $\Lambda=(\Lambda^{m,k})_{k \leq m}$ be a sequence of OWA
  operators such that $\Lambda_{m,k} = (\lambda_1^{m,k}, \ldots,
  \lambda_k^{m,k})$ has dimension $k$. Let $\gamma =
  (\gamma_{m,k})_{k \le m}$ be a family of single-winner scoring
  functions. %
  Then, $\gamma$ and $\Lambda$ define a family $f =
  (f_{m,k})_{k \le m}$ of committee scoring functions such that for
  each $(i_1, \ldots, i_k) \in [m]_k$ we have:
  \[
    f_{m,k}(i_1, \ldots, i_k) = \textstyle\sum_{t=1}^{k}\lambda_t^{m,k} \gamma_{m,k}(i_t).
  \]
  We refer to committee scoring rules $\calR_f$ defined through $f$ in
  this way as OWA-based.
\end{definition}

It is known that weakly separable, representation-focused, and
top-$k$-counting rules are OWA-based. The first class is defined using
OWA operators $(1, \ldots, 1)$, the second one uses OWA operators
$(1,0, \ldots,0)$, and the last one contains rules that use $k$-Approval single-winner
scoring functions and any OWA operator (the argument that shows
this is due to Faliszewski et
al.~\cite[Proposition 3]{fal-sko-sli-tal:c:top-k-counting} and requires a bit more
effort than for the previous two classes). As a corollary to the
preceding propositions, we get the following.

\begin{corollary}
  Each of the classes of separable, top-$k$-counting, and
  representation-focused rules is strictly contained in the class of
  OWA-based rules.
\end{corollary}
\begin{proof}
  Containment follows from the paragraph above.
  Strictness follows as we have Bloc as the unique rule in the intersection of
  top-$k$-counting and weakly separable; SNTV as the unique rule in the
  intersection of weakly separable and representation-focused; and
  $\alpha_k$-CC as the unique rule in the intersection of
  top-$k$-counting and representation-focused:
    it follows that Bloc is not representation-focused; SNTV is not top-$k$-counting; and
  $\alpha_k$-CC is not weakly separable.  We get the claim by noticing that
  Bloc, SNTV, and $\alpha_k$-CC, are all OWA-based.
\end{proof}

Naturally, there are also OWA-based rules that do not belong to any of
the above-mentioned classes. For example, this is the case for
$\alpha_t$-PAV rules (provided that the parameter $t$ is not equal to
the committee size $k$, e.g., if it is fixed as a constant) or for the
related $q$-HarmonicBorda rules (the $q$-HB rules), defined by the
following scoring functions ($q \in \mathbb{R}_{+}$ is a parameter):
\begin{align*}
  f^\qbordapav_{m,k}(\row ik) = \beta_m(i_1) +
  \textstyle\frac{1}{2^q}\beta_m(i_2) + \frac{1}{3^q}\beta_m(i_3) +
  \cdots + \frac{1}{k^q}\beta_m(i_k) \text{.}
\end{align*}
The $q$-HarmonicBorda rules were introduced by Faliszewski et
al.~\cite{fal-sko-sli-tal:c:paths}, who were looking for various means
of achieving a compromise between the $k$-Borda rule and the
Chamberlin--Courant rule ($0$-HB is $k$-Borda, and as $q$ becomes
larger and larger, $q$-HB becomes more and more similar to
$\beta$-CC).

\begin{proposition}\label{pro:owa}
  Neither $\alpha_t$-PAV nor $q$-HB is weakly separable, nor
  representation-focused, nor top-$k$-counting, for any choice of
  constants $t \in \naturals$ and $q \in \reals_+$.
\end{proposition}

To prove Proposition~\ref{pro:owa} it suffices to show that the
committee scoring functions of these rules cannot be expressed as linear
transformations of weakly separable, representation-focused, and
top-$k$-counting scoring functions, and invoke Lemma~\ref{lem:unique}.
We omit the details of this simple but somewhat tedious task.

Skowron et al.~\cite{sko-fal-lan:c:collective} have shown that
OWA-based rules are typically $\np$-hard to compute (with the clear
exception of, e.g., weakly separable rules and the Perfectionist rule). They
have also linked the properties of the OWA operators with the ability
to approximate the rules (generally speaking, if the OWA operators for
a given rule are non-increasing then there are polynomial-time
constant-factor approximation algorithms for this rule, and otherwise
they are typically missing\footnote{However, there are exceptions. For
  example, viewed as an OWA-based rule, Perfectionist uses OWA
  operators $(0, \ldots, 0,1)$ but still is polynomial-time
  computable. This is because, as a top-$k$-counting rule,
  Perfectionist uses a very restrictive single-winner scoring
  function, and is not captured by the results of Skowron et
  al.~\cite{sko-fal-lan:c:collective}.}).

\subsection{Decomposable Rules}
We introduce the following class that naturally generalizes the class
of OWA-based rules and resort to our T-shirt store example to help the
reader rationalize it.

\begin{definition}\label{def:decomposable-csr}
  Let $\gamma = (\gamma^{(t)}_{m,k})_{t \le k \le
    m}$ %
  be a family of single-winner scoring functions. %
  These functions define
  a family of committee scoring functions $f = (f_{m,k})_{k \le m}$
  such that for each committee position $(i_1, \ldots, i_k) \in [m]_k$
  we have:
  \[
    f_{m,k}(i_1, \ldots, i_k) = \textstyle\sum_{t=1}^{k}\gamma^{(t)}_{m,k}(i_t).
  \]
  We refer to committee scoring rules $\calR_f$ defined through $f$ in
  this way as {decomposable}.
\end{definition}

At first glance, decomposable rules seem very similar to the weakly
separable ones. The difference is that for fixed $m$ and $k$ and two
different values $t$ and $t'$, for decomposable rules the functions
$\gamma^{(t)}_{m,k}$ and $\gamma^{(t')}_{m,k}$ can be completely
different. It is apparent that OWA-based rules are decomposable.  
We will see that this containment is strict.

\begin{example}\label{example:shirts:decomposable}
  Let us recall from Example~\ref{example:shirts} that a customer
  considers a T-shirt to be ``good enough'' if it is among the best
  $10\%$ of all shirts and let us say that a shirt is ``great'' if it
  is among the top $1\%$ of all shirts.
  A customer buys two ``great'' T-shirts, or one ``at least good
  enough'' T-shirt (if there are no two ``great'' T-shirts on
  display). Naturally, the customer picks the best T-shirt(s) he can
  find (respecting the above constraints).  If $i_1, \ldots, i_k$ are
  the positions (in the customer's preference order) of the T-shirts
  that the store puts on display, then the number of T-shirts he or
  she buys is given by function:
   \[
     f_{m,k}(i_1, \ldots, i_k) = \alpha_{0.10m}(i_1) + \alpha_{0.01m}(i_2).
   \]
   Thus, to maximize its revenue, the store should find a winning
   committee for the election where the T-shirts are the candidates,
   the voters are the customers, and where we use decomposable
   committee scoring rule $\calR_f$ based on $f = (f_{m,k})_{k \le
     m}$.
\end{example}

We refer to decomposable rules defined through committee scoring
functions of the form
\[
  f_{m,k}(i_1, \ldots, i_k) = \lambda^k_1 \alpha_{t_{m,k,1}}(i_1) + \cdots +\lambda^k_k \alpha_{t_{m,k,k}}(i_k),
\]
where $\Lambda_k = (\lambda^k_1, \ldots, \lambda^k_k)$ are OWA
operators and $t_{m,k,1}, \ldots, t_{m,k,k}$ are sequences of
integers from $[m]$, as \emph{multithreshold} rules (we put no constraints on
$t_{m,k,1}, \ldots, t_{m,k,k}$; both increasing and decreasing
sequences are natural).

\begin{proposition}\label{prop:multithreshold-not-owa}
  The committee scoring rule defined through the multithreshold functions
  $f_{m,k}(i_1, \ldots, i_k) = \alpha_{p_1}(i_1) + \alpha_{p_2}(i_2)$,
  for $p_1, p_2 \in \{2, \ldots, m-k-2\}$, $p_1 > p_2+1 \geq 3$, is not
  OWA-based.
\end{proposition}
\begin{proof}
  Let us fix $p_1$, $p_2$, $m$, and $k$ that satisfy the requirements
  from the statement of the theorem.  For the sake of contradiction,
  assume that our multithreshold function is OWA-based.  By
  Lemma~\ref{lem:unique} we infer that there exist a  committee scoring function
  $g_{m,k}$ of the form:
  \begin{align*}
    g_{m,k}(i_1, \ldots, i_k) = \lambda_1 \gamma(i_1) + \lambda_2 \gamma(i_2) \text{,}
  \end{align*}
  where $\lambda_1, \lambda_2 \in \reals$ are two numbers and $\gamma$
  is a single-winner scoring function, such that for each committee
  position $I = (i_1, \ldots, i_k)$ it holds that $f_{m,k}(I) =
  g_{m,k}(I)$; this follows because, by Lemma~\ref{lem:unique}, the
  OWA-based committee scoring functions for our rule have to depend on
  $i_1$ and $i_2$ only, and by applying appropriate linear
  transformations, we can assume that these functions equal $f_{m,k}$.

  Let us now consider two committee positions $I' = (p_2, p_1+1,
  \ldots)$ and $I'' = (p_2, p_1, \ldots)$.  We see that:
  \begin{align*}
  f_{m,k}(I') - f_{m,k}(I'') &= \big(\alpha_{p_1}(p_2) + \alpha_{p_2}(p_1+1)\big) - \big(\alpha_{p_1}(p_2) + \alpha_{p_2}(p_1)\big) = \alpha_{p_2}(p_1+1) - \alpha_{p_2}(p_1) = 0 \text{,}
  \end{align*}
  and, thus, it must also be the case that:
  \begin{align*}
  g_{m,k}(I') - g_{m,k}(I'') = \big(\lambda_1\gamma(p_2) + \lambda_{2}\gamma(p_1+1)\big) - \big(\lambda_{1}\gamma(p_2) + \lambda_{2}\gamma(p_1)\big) = \lambda_2\big(\gamma(p_1+1) - \gamma(p_1)\big) = 0  
  \end{align*}
  On the other hand, for committee positions $J' = (p_1+1,p_1+2,
  \ldots)$ and $J'' = (p_1, p_1+2, \ldots)$ we have:
  \begin{align*}
  f_{m,k}(J') - f_{m,k}(J'') = \big(\alpha_{p_1}(p_1+1) + \alpha_{p_2}(p_1+2)\big) - \big(\alpha_{p_1}(p_1) + \alpha_{p_2}(p_1+2)\big) < 0 
  \end{align*}
  and, consequently:
  \begin{align*}
    g_{m,k}(J') - g_{m,k}(J'') &= \big(\lambda_1\gamma(p_1+1) +
    \lambda_{2}\gamma(p_1+2)\big) - \big(\lambda_{1}\gamma(p_1) +
    \lambda_{2}\gamma(p_1+2)\big) \\ &= \lambda_1 \big(\gamma(p_1+1) -
    \gamma(p_1)\big) < 0.
  \end{align*}
  Since we have both $\lambda_2\big(\gamma(p_1+1) - \gamma(p_1)\big) =
  0$ and $\lambda_1 \big(\gamma(p_1+1) - \gamma(p_1)\big) < 0$, we
  conclude that $\lambda_2 = 0$.  However, for committee positions $L'
  = (p_2-1, p_2+1, \ldots)$ and $L'' = (p_2-1, p_2, \ldots)$ we have:
  \begin{align*}
  f_{m,k}(L') - f_{m,k}(L'') = \big(\alpha_{p_1}(p_2-1) + \alpha_{p_2}(p_2+1)\big) - \big(\alpha_{p_1}(p_2-1) + \alpha_{p_2}(p_2)\big) < 0
  \end{align*}
  and:
  \begin{align*}
  g_{m,k}(L') - g_{m,k}(L'') &= \big(\lambda_1\gamma(p_2-1) +
    0 \cdot \gamma(p_2+1)\big) - \big(\lambda_{1}(p_2-1) +
    0 \cdot (p_2)\big) = 0 \text{,}
  \end{align*}
  which is a contradiction  and completes the proof. 
\end{proof}

We generally expect decomposable rules to be $\np$-hard, but even
among these rules there are polynomial-time computable rules (that are
not OWA-based). For example, in their discussion of top-$k$-counting
rules, Faliszewski et al.~\cite{fal-sko-sli-tal:c:top-k-counting}
mention a multithreshold rule that uses scoring functions that mix
SNTV and Perfectionist:
\begin{align*}
  f_{m,k}^{\sntv + \perf}(i_1, \ldots, i_k) = f_{m,k}^{\sntv}(i_1,
  \ldots,i_k) + f_{m,k}^{\perf}(i_1, \ldots,i_k) = \alpha_1(i_1) +
  \alpha_k(i_k) \text{.}
\end{align*}
Briefly put, each winning committee under this rule is either an SNTV
winning committee or is ranked on top $k$ positions by some voter, and
it suffices to check all such possibilities (thus, e.g., it is
possible to compute some winning committee in polynomial time). One
can show that this rule is not OWA-based using the same approach as in
Proposition~\ref{prop:multithreshold-not-owa}.

\subsection{Beyond Decomposable Rules}

Naturally, there are also committee scoring rules that go beyond the
class of decomposable rules. Below we provide two examples, starting
with one inspired by our T-shirt store.

\begin{example}
  In this example, the store does not want to maximize its direct
  revenue (i.e., the number of T-shirts sold), but the number of happy
  customers (in hope of increased future revenue). Let us say that a
  customer is happy if he or she finds at least two ``good enough''
  T-shirts or at least one ``great'' T-shirt (recall that ``at least
  good enough'' shirts are among top $10\%$ of all available ones, and
  ``great'' shirts are among the top $1\%$). Then the store should use
  the committee scoring function
 $$
   f_{m,k}(i_1, \ldots, i_k) = \max( \alpha_{0.01m}(i_1),   \alpha_{0.10m}(i_2)).
 $$
\end{example}

We refer to multithreshold rules with summation replaced by the $\max$
operator as \emph{max-threshold} rules. Using an approach similar to that
from Proposition~\ref{prop:multithreshold-not-owa}, one can show that
there are max-threshold rules that are not decomposable (we omit
details).

In their search for rules between $k$-Borda and $\bordacc$,
Faliszewski et al.~\cite{fal-sko-sli-tal:c:paths} introduced the class
of $\ell_p$-Borda rules, based on the following scoring functions ($p
\geq 1$ is a parameter):
\begin{align*}
  f^\ellpborda_{m,k}(\row ik) = \sqrt[p]{\beta_m^p(i_1) + \cdots + \beta_m^p(i_k)} \text{.}
\end{align*}
While the motivation for these rules is the same as for the
$q$-HarmonicBorda rules, they behave quite differently (see the work
of Faliszewski et al.~\cite{fal-sko-sli-tal:c:paths} for a detailed
discussion). 

\begin{corollary}
  There are committee scoring rules that are not decomposable.
\end{corollary}

Throughout the rest of the paper, we will not venture outside the
class of decomposable rules. However, the above two examples show that
there are interesting rules there that also deserve to be studied
carefully.

\section{Axiomatic Properties of Committee Scoring Rules}
\label{sec:axiomatic_properties}

After exploring the universe of committee scoring rules from a
syntactic (structural) perspective, we now consider axiomatic
properties of the observed classes.
Specifically, we will use two types of monotonicity
notions---non-crossing monotonicity (together with its relaxations)
and committee enlargement monotonicity---to characterize several of
the classes and to gain insights regarding some others.
Indeed, various monotonicity concepts have long been used in social
choice (with Maskin monotonicity~\cite{maskin1999nash} being perhaps the
most important example) and we follow this tradition.

%
%
%

%
%
%
%
%

\subsection{Non-crossing Monotonicity and Its Relaxations}

\label{sec:noncrossing_monotonicity}

Elkind et al.~\cite{elk-fal-sko-sli:j:multiwinner-properties}
introduced two monotonicity notions for multiwinner rules, namely
candidate monotonicity (recall Section~\ref{sec:basic}) and
non-crossing monotonicity. In the former, we require that if we shift
forward a candidate from a winning committee in some vote, then this
candidate still belongs to some winning committee after the shift, but
possibly to a different one.
In the latter monotonicity notion, we require that the whole committee
remains winning, but we forbid shifts were members of the winning
committee pass each other (i.e., after a shift none of the committee
members gets worse and some get better). More formally, we have the
following definition.

\begin{definition}[Elkind et al.~\cite{elk-fal-sko-sli:j:multiwinner-properties}]\label{def:noncrossing_monotonicity}
  A multiwinner rule $\calR$ is non-crossing monotone if
  for each election $E = (C,V)$ and each $k \in [|C|]$ the following
  holds: if $c \in W$ for some $W \in \calR(E,k)$, then for each $E'$
  obtained from~$E$ by shifting $c$ forward by one position in some
  vote without passing another member of $W$, we still have  $W \in
  \calR(E',k)$.
\end{definition}

Elkind et al.~\cite{elk-fal-sko-sli:j:multiwinner-properties} have
shown that weakly separable rules are non-crossing monotone, and we
will now show that the converse is also true.  However, before we
proceed to the proof, we introduce the following notation (that will
also be useful in further analysis):
\begin{enumerate}
\item[] Consider an arbitrary number of candidates $m$ and a size of
  committee $k \in [m]$.  For each $t \in [k]$ and $p \in
  [m]$, %
  let $P_{m, k}(t, p)$ be the set of committee positions from $[m]_k$
  that have their \hbox{$t$-th} element equal to $p$ and such that
  they do not include position $p-1$.  We set $P_{m, k}(p) =
  \bigcup_{t \leq k} P_{m, k}(t, p)$.
\end{enumerate}
For example, if $m=5$ and $k=3$,
then
$P_{5,1}(1, 4) = \emptyset$,
$P_{5,3}(2, 4) = \{(1, 4, 5), (2, 4, 5)\}$,
$P_{5,3}(3, 4) = \{(1, 2, 4)\}$,
and
$P_{5,3}(4) = P_{5,3}(1,4) \cup P_{5,3}(2,4) \cup P_{5,3}(3,4) = \{(1, 4, 5), (2, 4, 5), (1, 2, 4)\}$.

Intuitively, $P_{m,k}(t, p)$ is a collection of
committee positions  in which the $t$-th committee member
stands on position $p$ and where shifting him or her %
without passing another committee
member is possible.  
Similarly, $P_{m,k}(p)$ is a collection of committee positions 
in which there is \emph{some} committee member on position $p$ and it
is possible to shift him to position $p-1$ without passing another
committee member.

\begin{theorem}\label{thm:weaklysep}
  Let $\calR_f$ be a committee scoring rule.
$\calR_f$ 
  is non-crossing monotone if and only if it is weakly
  separable. %
\end{theorem}
\begin{proof}
  Let $\calR_f$ be a committee scoring rule defined through a family
  $f=(f_{m,k})_{k\le m}$ of scoring functions $f_{m,k}\colon [m]_k\to
  \reals$.  Due to the results of Elkind et
  al.~\cite{elk-fal-sko-sli:j:multiwinner-properties}, it suffices to
  show that if $\calR_f$ is non-crossing monotone then it is weakly
  separable. So let us assume that $\calR_f$ is non-crossing monotone.

  Let us fix the number of candidates $m$ and the committee size $k
  \in [m]$. Let $E = (C,V)$ be an election with candidate set $C = \{c_1,
  \ldots, c_m\}$ and collection of voters $V = (v_1, \ldots, v_{m!})$,
  with one voter for each possible preference order. By symmetry, every
  size-$k$ subset $W$ of $C$ is a winning committee %
  under $\calR_f$.

  Consider an arbitrary integer $p \in \{2, \ldots, m\}$, two
  arbitrary (but distinct) committee positions $I = (i_1, \ldots,
  i_k)$ and $J = (j_1, \ldots, j_k)$ from $P_{m,k}(p)$, and an
  arbitrary vote $v$ from the election.  Let $C(I)$ be the set of
  candidates that $v$ ranks at positions $i_1, \ldots, i_k$, and let
  $C(J)$ be defined analogously for the case of~$J$. Let $E'$ be the
  election obtained by shifting in $v$ the candidate currently in
  position $p$ one position up. Finally, let $I'$ and $J'$ be
  committee positions obtained from $I$ and $J$ by replacing the
  number $p$ with $p - 1$ (it is possible to do so as $I$ and $J$ are
  both from $P_{m,k}(p)$). %

  Since, by assumption, $\calR_f$ is non-crossing monotone, it must be
  the case that $C(I)$ and $C(J)$ are winning committees under
  $\calR_f$ also in election $E'$. The difference of the scores of committee $C(I)$ in elections
  $E'$ and $E$ is $f_{m,k}(I') - f_{m,k}(I)$, and the difference of the scores of committee
  $C(J)$ in $E'$ and $E$ is $f_{m,k}(J') - f_{m,k}(J)$. It must be the case that:
  \begin{align*}
     f_{m,k}(I') - f_{m,k}(I) = f_{m,k}(J') - f_{m,k}(J)\ge 0\text{.}
  \end{align*}
  However, since the choice of $p$ and the choices of $I$ and $J$
  within $P_{m,k}(p)$ were completely arbitrary, it must be the case
  that there is a function $h_{m,k}$ such that for each $p \in \{2,
  \ldots, m\}$, each sequence $U \in P_{m,k}(p)$, and each committee
  position $U'$ obtained from $U$ by replacing position $p$ with
  $p-1$, we have:
  \begin{align*}
    h_{m,k}(p-1) = f_{m,k}(U') - f_{m,k}(U) 
  \end{align*}
  and the values of $ h_{m,k}$ are non-negative.

  Our goal now is to construct a single-winner scoring  function $\gamma_{m,k}$
  such that for each committee position $(\ell_1, \ldots, \ell_k) \in [m]_k$ it holds that:
  \begin{align*}
    f_{m,k}(\ell_1, \ldots, \ell_k) = \gamma_{m,k}(\ell_1) +\gamma_{m,k}(\ell_2)+ \cdots + \gamma_{m,k}(\ell_k) \text{.}
  \end{align*}
  We define  $\gamma_{m,k}$ by requiring that
  (a) for
  each $p \in \{2, \ldots, m\}$, we have $\gamma_{m,k}(p-1)-\gamma_{m,k}(p) = h_{m,k}(p-1)$
  (so $ \gamma_{m,k}$ is a non-increasing function), and 
  (b) 
  $\gamma_{m,k}(m) $ is such that $\gamma_{m,k}(m) +
  \gamma_{m,k}(m-1) + \ldots + \gamma_{m,k}(m-(k-1)) =
  f_{m,k}(m-(k-1), \ldots, m-1, m)$ (so that $\gamma_{m,k}$ indeed
  correctly describes the $f_{m,k}$-score of the committee ranked at the $k$
  bottom positions as a sum of the scores of the candidates).

  We fix some committee position $(\ell_1, \ldots, \ell_k)$ from
  $[m]_k$.
  We know that, due to the choice of $\gamma_{m,k}(m)$, for $R = (r_1, \ldots, r_k) = (m-k+1, \ldots, m)$ it does
  hold that $f_{m,k}(r_1, \ldots, r_k) = \gamma_{m,k}(r_1) + \cdots +
  \gamma_{m,k}(r_k)$. 
  Now we can see that this property also holds for $R' = (r_1-1, r_2,
  \ldots, r_k)$. The reason is that
  \begin{align*}
  \gamma_{m,k}(m-k) - \gamma_{m,k}(m-k+1) = h_{m,k}(m-k) = f_{m,k}(R') - f_{m,k}(R) \text{.}
  \end{align*}
  Thus, for $R'$, we have $f_{m,k}(R') = \gamma_{m,k}(r_1-1) + \gamma_{m,k}(r_2) + \cdots + \gamma_{m,k}(r_k)$.
  We can proceed in this way, shifting the top member of the
  committee up by sufficiently many positions, to obtain $R'' = (\ell_1,
  r_2, \ldots, r_k)$ and (by the same argument as above) have:
  \begin{align*}
  f_{m,k}(R'') = \gamma_{m,k}(\ell_1) + \gamma_{m,k}(r_2) + \cdots + \gamma_{m,k}(r_k) \text{.}
  \end{align*}
  Then we can do the same to position $r_2$, and keep decreasing it
  until we get $\ell_2$. Then the same for the third position, and so on,
  until the $k$-th position. Finally, we get:
\begin{align*}
  f_{m,k}(\ell_1, \ldots, \ell_k) = \gamma_{m,k}(\ell_1) + \cdots + \gamma_{m,k}(\ell_k) \text{.}
\end{align*}
  This 
  proves our claim and 
 completes the proof.
\end{proof}

Non-crossing monotonicity is particularly natural when we seek
committees of individually excellent candidates (for example, when we
seek finalists of a competition or where we are interested in some shortlisting
tasks~\cite{elk-fal-sko-sli:j:multiwinner-properties,fal-sko-sli-tal:b:multiwinner-trends}). Indeed,
if we have a committee $W$ where we view each member as good enough to
be selected, and one of the members of $W$ improves its performance
without hurting the performance of any of the others, then it is
perfectly natural to expect that all members of $W$ are still good
enough to be selected.
Theorem~\ref{thm:weaklysep} justifies axiomatically that if we are
looking for a committee scoring rule for selecting individually
excellent candidates then we should look within the class of weakly
separable rules.
In fact, Elkind et al.~\cite{elk-fal-sko-sli:j:multiwinner-properties}
pointed out that we should focus on separable rules only, and we will
provide axiomatic justification for this view in
Section~\ref{sec:com_mon_sep_rules}.

\subsubsection{Prefix Monotonicity and Decomposable Rules}\label{sec:other_noncrossing_monotonicity}

Based on the idea of non-crossing monotonicity, we can define other
similar notions. In this section we introduce and discuss one of them,
which we call {\em prefix monotonicity}.  Intuitively, if a rule
satisfies the prefix monotonicity condition, then shifting forward a
group of highest-ranked members of a winning committee within a given
vote never prevents this committee from winning.

\begin{definition}%
\label{def:prefix_monotonicity}
A multiwinner rule $\calR$ satisfies \emph{$t$-prefix monotonicity},
$0 \leq t \leq k$, if for each election $E = (C,V)$ and each committee
size $k$, $t \leq k \leq |C|$ the following holds: For every $W \in
\calR(E,k)$, and every $E'$ obtained from~$E$ by shifting in some vote
the top-ranked $t$ members of $W$ (according to this vote), then we have
that $W \in \calR(E',k)$. We say that $\calR$ satisfies \emph{prefix
  monotonicity} if it satisfies $t$-prefix monotonicity for every $t
\in \naturals$.\footnote{Note that $0$-prefix monotonicity is an empty concept; as such, every rule satisfies it.}
\end{definition}

Prefix monotonicity is a relaxation of non-crossing monotonicity and,
in consequence, all weakly separable rules satisfy it.  In the
remaining part of this section we will show that only decomposable
rules can be prefix-monotone (and mostly, though not only, those based
on convex functions; we will explain this in the further part of this
section).  Before we prove this statement, let us first prove one more
technical lemma, which will allow us to reuse some of the reasoning
later on. The lemma uses the same high-level idea as the first part of
the proof of Theorem~\ref{thm:weaklysep}, yet it is more involved and
differs in a number of details.  (Recall that $P_{m, k}(t, p)$ used in
the statement of the lemma was defined right before
Theorem~\ref{thm:weaklysep}.)

\begin{lemma}\label{lem:t_prefix_monotonicity}
  Let $\calR_f$ be a committee scoring rule and let $t$ be an integer
  such that for each $x \in [t]$ this rule is $x$-prefix
  monotone.
  Then, for every number of candidates $m$ and size of the committee~$k$,
  there exists a function $h_t$ such that for each $p\in [m]$,
  each $U \in P_{m, k}(t, p)$, with $p\ge t$, and committee position
  $U'$ obtained from $U$ by replacing position $p$ with $p-1$, we
  have:
\begin{equation*}
\label{functionh_t}
      h_t(p-1) = f(U') - f(U) \geq 0 \textrm{.}
\end{equation*}
\end{lemma}

\begin{proof}
  Consider an arbitrary integer $p \in [m]$ %
  and two arbitrary (but distinct) committee positions $I = (i_1,
  \ldots, i_k)$ and $J = (j_1, \ldots, j_k)$ from $P_{m,k}(t, p)$,
  such that $I$ and $J$ have the first $t$ elements equal. Let $I'$
  and $J'$ be the committee positions obtained from $I$ and $J$,
  respectively, by replacing the element $p$ with $p-1$ (by the choice
  of $I$ and $J$, it is possible to do so). Let $I_1$ and $J_1$ be the
  committee positions obtained from $I$ and $J$, respectively, by
  increasing every element with value lower than $p$ by one (in particular, when
  $t=1$ we have $I_1=I$ and $J_1=J$). The way the sequences $I'$ and
  $I_1$ are constructed from $I$ is depicted in
  Figure~\ref{fig:shifting} (in this example $t=4$) and is presented
  below, also for $J'$ and $J_1$ ($i_t = p$, $i_{t-1} \leq p-2$, $j_1
  = i_1, \ldots, j_{t} = i_{t}$):
  \begin{align*}
    &I = (i_1, i_2 \ldots, i_{t-1}, p, i_{i+1}, \ldots, i_k), &  & J = (i_1, i_2 \ldots, i_{t-1}, p, j_{i+1}, \ldots, j_k), \\
    &I' = (i_1, i_2 \ldots, i_{t-1}, p-1, i_{i+1}, \ldots, i_k), &  & J' = (i_1, i_2 \ldots, i_{t-1}, p-1, j_{i+1}, \ldots, j_k), \\
    &I_1 = (i_1+1, i_2+1 \ldots, i_{t-1}+1, p, i_{i+1}, \ldots, i_k), &  & J_1 = (i_1+1, i_2+1 \ldots, i_{t-1}+1, p, j_{i+1}, \ldots, j_k).
  \end{align*}

  \begin{figure}[tb]
    \begin{center}
	\includegraphics[scale=0.60]{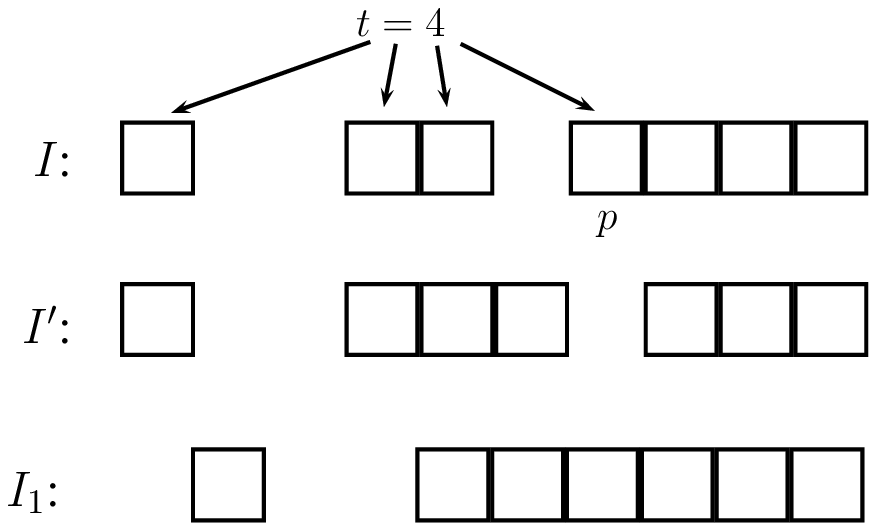}
      \end{center}
      \caption{An example showing how the sequences of positions $I$,
        $I'$ and $I_1$ from the proof of
        Lemma~\ref{lem:t_prefix_monotonicity} are related.}
      \label{fig:shifting}
    \end{figure}

    As in the proof of Theorem~\ref{thm:weaklysep}, we construct an
    election $E = (C,V)$ with candidate set $C = \{c_1, \ldots, c_m\}$
    and $m!$ voters $v_1, \ldots, v_{m!}$, one for each possible
    preference order. By symmetry, every size-$k$ subset $W$ of $C$ is
    a winning committee of $E$ under $\calR_f$. Further, consider an
    arbitrary vote $v$ from the election; let $C(I_1)$ and $C(J_1)$ be
    the committees that $v$ ranks on positions $I_1$ and $J_1$,
    respectively. As all other committees, $C(I_1)$ and $C(J_1)$ are
    winning in $E$. Let us shift in $v$ by one position forward each
    candidate from $C(I_1)$ that stands on a position with value lower
    than $p$. After such an operation, committee $C(I_1)$ will have
    position $I$ and committee $C(I_1)$ will have position $J$. Since,
    by assumption, $\calR_f$ is $(t-1)$-prefix-monotone, and exactly
    $t-1$ candidates have changed positions, it must be the case that
    $C(I_1)$ and $C(J_1)$ are still winning under $\calR_f$. 
    It must be the case that:
    \begin{align*}
      f(I) - f(I_1) = f(J) - f(J_1)\ge 0\textrm{.}
    \end{align*}
    By a similar reasoning, using the fact that $\calR_f$ is
    $t$-prefix-monotone, we also conclude that:
    \begin{align*}
      f(I') - f(I_1) = f(J') - f(J_1)\ge 0\textrm{.}
    \end{align*}
    From the two above equalities we get that (the final inequality follows because
    $J'$ dominates $J$):
    \begin{equation}\label{eq:pm:1}
      f(I') - f(I) = f(J') - f(J) \geq 0 \textrm{.}
    \end{equation}

    Recall that in the above equality $I$ and $J$ have the first $t-1$
    elements equal. We would like to obtain the same relation even if
    the prefixes of $I$ and $J$ differ. Thus, now we will show how to
    change one element in the prefix of $I$ and $I'$ to an element
    different by one, so that the equality still holds. By repeating
    this operation sufficiently many times, we can conclude that the
    equality does not depend on the prefix of $I$. For the sake of
    concreteness, we will show how to change $i_{t-1}$ to $i_{t-1}+1$
    in the prefix of $I$ (this assumes that $i_{t-1}+1 < p-1$). A
    change of any other element in the prefix can be performed
    analogously. We proceed as follows.  Let us define:
  \begin{align*}
    &I_{\mathrm{new}} = (i_1, i_2 \ldots, i_{t-1}+1, p, i_{i+1}, \ldots, i_k), &  & L_{\mathrm{new}} = (i_1, i_2 \ldots, i_{t-1}+1, p-1, i_{i+1}, \ldots, i_k), \\
    &I_{\mathrm{new}}' = (i_1, i_2 \ldots, i_{t-1}, p, i_{i+1}, \ldots, i_k), &  & L_{\mathrm{new}}' = (i_1, i_2 \ldots, i_{t-1}, p-1, i_{i+1}, \ldots, i_k).
  \end{align*}
  In particular, observe that $I_{\mathrm{new}}' = I$ and that
  $L_{\mathrm{new}}' = I'$.  Similarly as before, by using
  $(t-1)$-prefix-monotonicity and $(t-2)$-prefix-monotonicity, we
  obtain that:
    \begin{equation}\label{eq:pm:2}
      f(I_{\mathrm{new}}') - f(I_{\mathrm{new}}) = f(L_{\mathrm{new}}') - f(L_{\mathrm{new}}) \geq 0 \textrm{.}
    \end{equation}
    Adding  inequalities~\eqref{eq:pm:1} and~\eqref{eq:pm:2}, we get:
    \begin{align*}
      f(I') - f(I) + f(I_{\mathrm{new}}') - f(I_{\mathrm{new}}) =  f(J') - f(J)  + f(L_{\mathrm{new}}') - f(L_{\mathrm{new}}) \textrm{.}
    \end{align*}
    which is equivalent to:
    \begin{align*}
     f(L_{\mathrm{new}}) - f(I_{\mathrm{new}}) =  f(J') - f(J)   \textrm{.}
    \end{align*}
    However, we can see that $L_{\mathrm{new}}$ and $I_{\mathrm{new}}$
    are simply $I'$ and $I$ where one element of the prefix,
    $i_{t-1}$, is replaced with $i_{t-1}+1$. By our previous
    discussion, it follows that we can prove that $f(I') - f(I) =
    f(J') - f(J)$ even if $I$ and $J$ have different prefixes.

    Since the choice of $p$, $I$, and $J$ (within $P_{m,k}(t, p)$) is
    completely arbitrary, it must be the case that for each $t$ there
    exists a function $h_t$ such that for each $p \in \{t+1, \ldots,
    m\}$, each sequence $U \in P_{m,k}(t, p)$, and each sequence $U'$
    obtained from $U$ by replacing position $p$ with $p-1$, we have:
    \begin{equation*}
      h_t(p-1) = f(U') - f(U) \geq 0. 
    \end{equation*}
    The final inequality follows from equation~\eqref{eq:pm:1}.
\end{proof}

We are ready to show that only decomposable rules can satisfy prefix-monotonicity.

\begin{theorem}\label{thm:prefixMonAndDecomposable}
  Let $\calR_f$ be a committee scoring rule. %
  If $\calR_f$ is prefix-monotone then it must be decomposable.
\end{theorem}

\begin{proof}
  Let $f = (f_{m,k})_{k \le m}$ be a family of committee scoring
  functions such that $\calR_f$ is prefix-monotone. Let us fix the
  number of candidates $m$ and the committee size $k$.  For each $t
  \in [k]$, let $h_t$ be the function constructed in
  Lemma~\ref{lem:t_prefix_monotonicity}.

  Our goal is to provide single-winner scoring functions
  $\gamma^{(1)}_{m,k}, \ldots, \gamma^{(k)}_{m,k}$ such that for each
  committee position $(\ell_1, \ldots, \ell_k)$ we have:
  \begin{equation}\label{eq:fpm:1}
  f_{m,k}(\ell_1, \ldots, \ell_k) = \gamma^{(1)}_{m,k}(\ell_1) +\gamma^{(2)}_{m,k}(\ell_2)+
    \cdots + \gamma^{(k)}_{m,k}(\ell_k) \textrm{.}
  \end{equation}
  To this end, for each $t \in [k]$,  we define
  $\gamma^{(t)}_{m,k}: \{t, \ldots, m - k + t\}
  \rightarrow \reals$ so that:\footnote{Formally, $\gamma_t$ must
    be defined on $[m]$ but it actually never has a chance to
    calculate values $\gamma_t(s)$, where $s<t$ or $s>m - k + t$, so
    these values of $\gamma_t$ can be chosen arbitrarily.}
  \begin{enumerate}
  \item The values $\gamma_k(m), \gamma_{k-1}(m - 1), \ldots,
    \gamma_1(m-(k-1))$ are such that $ f(m-(k-1), \ldots, m-1, m) =
    \gamma_k(m) + \gamma_{k-1}(m - 1) + \ldots + \gamma_1(m-(k-1))$
    (so equation~\eqref{eq:fpm:1} holds for the committee position
    where the candidates are ranked at the $k$ bottom positions).
  \item For each $p \in \{t+1, \ldots, m - k + t\}$, we have
    $\gamma^{(t)}_{m,k}(p-1)-\gamma^{(t)}_{m,k}(p) = h_t(p-1)$.
    (By Lemma~\ref{lem:t_prefix_monotonicity}, we have $h_t(p-1) \geq 0$, so 
    $\gamma^{(t)}_{m,k}$ is nonincreasing.)

  \end{enumerate}
  There may be many different ways to define functions
  $\gamma^{(1)}_{m,k}, \ldots, \gamma^{(k)}_{m,k}$ satisfying the
  above conditions and we choose one of them arbitrarily.

  To show that equation~\eqref{eq:fpm:1} holds, we use the same
  approach as in the second part of the proof of
  Theorem~\ref{thm:weaklysep}.  Specifically, we note that if equation
  \eqref{eq:fpm:1} holds for some committee position $R = (r_1, \ldots,
  r_k)$ and $R' = (r_1, \ldots, r_t-1, \ldots, r_k)$ also is a valid
  committee position for some $t \in [k]$, then (by definition of $h_t$) we have:
  \begin{align*}
    f_{m,k}(R') & = f_{m,k}(R) + h_t(r_t-1) \\ 
     & = \gamma^{(1)}_{m,k}(r_1) + \cdots + \gamma^{(t-1)}_{m,k}(r_{t-1})+ \bigg( \gamma^{(t)}_{m,k}(r_{t})+ h_t(r_t-1) \bigg) + \gamma^{(t+1)}_{m,k}(r_{t+1}) +  \cdots + \gamma^{(k)}_{m,k}(r_k)  \\
     & = \gamma^{(1)}_{m,k}(r_1) + \cdots + \gamma^{(t-1)}_{m,k}(r_{t-1})+ \gamma^{(t)}_{m,k}(r_{t}-1) + \gamma^{(t+1)}_{m,k}(r_{t+1}) +  \cdots + \gamma^{(k)}_{m,k}(r_k) \textrm{.} 
  \end{align*}
  Since equation~\eqref{eq:fpm:1} holds for committee position
  $(m-(k-1), \ldots, m)$, applying the above argument inductively
  proves that equation~\eqref{eq:fpm:1} holds for all committee
  positions.
  \end{proof}

  Theorem~\ref{thm:prefixMonAndDecomposable} states that
  decomposability is a necessary condition for a committee scoring
  rule to be prefix-monotone. However, as the following example shows,
  it is not sufficient.

\begin{example}
  The $k$-Approval Chamberlin--Courant rule ($\alpha_k$-CC), defined by committee
  scoring functions $f^\topkcc_{m,k}(i_1, \ldots, i_k) =
  \alpha_k(i_1)$, is a decomposable
  rule %
  that is not prefix-monotone. Indeed, consider $k=2$ and an election
  with four candidates $\{a, b, c, d\}$ that includes one vote for
  each possible ranking of these four candidates. This election
  contains $4! = 24$ votes and, in particular, vote $v\colon a \succ b
  \succ c \succ d$.  By the symmetry of the rule we see that for such
  election each committee is winning, including $W = \{b,c\}$ and $W'
  = \{c,d\}$. If $\alpha_k$-CC were prefix-monotone, then shifting $b$
  and $c$ by one position forward in $v$ (to obtain $b \succ c \succ a
  \succ d$) should keep $W$ winning. Doing so, however, does not change
  the score of $W$ and increases the score of $W'$, so $W$ no longer
  wins. This shows that $\alpha_k$-CC is not prefix-monotone.
\end{example} 

On the other hand,
if we assume that the single-winner scoring functions underlying a
decomposable rule are, in a certain sense, convex, then we obtain a
sufficient condition for this rule to be prefix-monotone.

\begin{proposition}\label{thm:prefixMonAndDecomposable2}
  Let $\calR_f$ be a decomposable committee scoring rule defined
  through a family of scoring functions
  $
    f_{m,k}(i_1, \ldots, i_k) = \gamma^{(1)}_{m,k}(i_1) +\gamma^{(2)}_{m,k}(i_2)+
    \cdots + \gamma^{(k)}_{m,k}(i_k) \textrm{,}
  $
  where $\gamma = (\gamma^{(t)}_{m,k})_{t \leq k \leq m}$ is a family
  of single-winner scoring functions.  A sufficient condition for
  $\calR_f$ to be prefix-monotone is that for each $m$ and each $k \in
  [m]$ we have that:
  \begin{enumerate}\item[(i)] for each $i \in [k]$ and each $p, p' \in [m-1]$, $p <
  p'$, it holds that:
  \begin{equation}
  \label{first_condition}
    \gamma^{(i)}(p) - \gamma^{(i)}(p+1) \geq \gamma^{(i)}(p') - \gamma^{(i)}(p'+1) \text{, and} 
  \end{equation} 
  \item[(ii)] for each $i, j \in [k]$, $j>i$, and each $p \in [m]$, $j \leq p < m - (k-i)$, 
  it holds that 
  \begin{equation}
  \label{second_condition}
    \gamma^{(i)}(p) - \gamma^{(i)}(p+1) \geq \gamma^{(j)}(p) - \gamma^{(j)}(p+1) \text{.}
  \end{equation}
  \end{enumerate}
\end{proposition}

Intuitively, condition (i) says that the functions in the  family $\gamma$
are convex, and condition (ii) says that, for each $m$ and $k$, if $i
< j$ then $\gamma^{(i)}_{m,k}$ decreases not faster than
$\gamma^{(j)}_{m,k}$. %

\begin{figure}[tb]
   \hspace{-3cm}\includegraphics[scale=0.60]{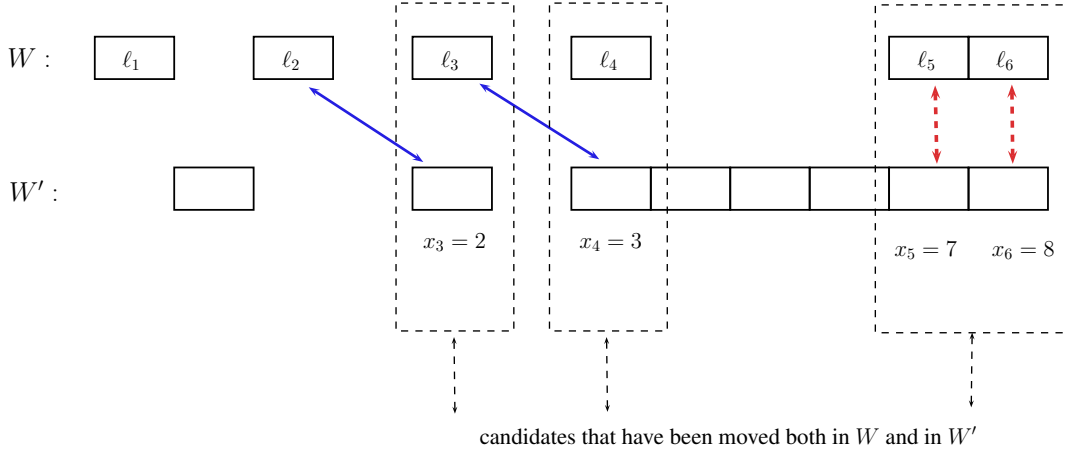}
   \caption{Illustration of the notation used in Proposition~\ref{thm:prefixMonAndDecomposable2}.}
   \label{fig:shifting2}
\end{figure}

\begin{proof}%
  Let $\calR_f$ be defined as in the statement of the proposition and
  fix the number of candidates $m$ and the committee size $k$.
  Consider an election $E$ where a committee $W$ is a winner. Let $j$
  be a number from $[k]$ and let $E'$ be an election obtained from $E$
  by shifting forward by one position each of the first $j$ members of
  $W$ in some vote $v$. We will show that $W$ is a winning committee
  in~$E'$. Let $(\ell_1, \ldots, \ell_k)$ be the committee position of
  $W$ in $v$ (in election $E$). In comparison with $E$, in~$E'$ the
  score of $W$ is increased by:
  \begin{align*}
    \sum_{t=1}^j \big(\gamma_t(\ell_t - 1) - \gamma_t(\ell_t)\big)
    \textrm{.}
  \end{align*}
  Let us now assess by how much the score of some other committee,
  $W'$, can increase. Let us fix $t\in[j]$, and let $c_t$ be the
  candidate standing at position $\ell_t$ in $v$ (in particular, $c_t\in W$). If $c_t \notin W'$,
  then shifting $c_t$ one position up has no positive effect on the
  score of $W'$. Consider the case when $c_t \in W'$. Let $x_t$ denote
  the position of $c_t$ within $W'$ according to $v$ (for instance, if
  $c_t$ is the most preferred among members of $W'$ in $v$, then $x_t =
  1$). This notation is illustrated in Figure~\ref{fig:shifting2}.
  Now, we consider two cases:
  
  \begin{description}
  \item[Case 1 ($\boldsymbol{x_t \geq t}$).] The condition~\eqref{second_condition} implies that:
    \begin{align*}
 \gamma_{x_t}(\ell_t - 1) - \gamma_{x_t}(\ell_t) \leq \gamma_t(\ell_t - 1) - \gamma_t(\ell_t) \text{.}
    \end{align*}
    Thus the increase of the score of $W'$ due to shifting $c_t$ one
    position up is not greater than the increase of the score of $W$
    due to shifting $c_t$ one position up. We assign $c_t$ in $W$ to
    $c_t$ in $W'$; this assignment is shown with a bold dashed arrow
    in Figure~\ref{fig:shifting2}) and, intuitively, it means that the
    increase of the score of $W$ due to shifting $c_t$ ``compensates
    for'' the increase of the score of $W'$ due to shifting the
    assigned candidate.

    \item[Case 2 ($\boldsymbol{x_t < t}$).] Now, we observe that due to \eqref{first_condition} we have:
    \begin{align*}
     \gamma_{x_t}(\ell_t - 1) - \gamma_{x_t}(\ell_t) \leq
      \gamma_{x_t}(\ell_{x_t} - 1) -
      \gamma_{x_t}(\ell_{x_t}).
    \end{align*}
    Thus the increase of the score of $W'$ due to shifting $c_t$ one
    position up is not greater than the increase of the score of $W$
    due to shifting  the candidate at position $\ell_{x_t} <
    \ell_t$, call such a candidate $c$, one position up. We assign
    $c_t$ in $W'$ to $c$ in $W$; this assignment is depicted with a
    solid arrow in Figure~\ref{fig:shifting2}.
  \end{description}

  From the above reasoning we see that for each $t\in[j]$ the increase
  of the score of $W'$ due to shifting $c_t$ one position up is no
  greater than the increase of the score of $W$ due to shifting some
  other candidate $c_{r}$ ($r \leq t$) one position up; in such case
  we say that $c_{r}$ is assigned to $c_t$ and that $c_{r}$
  compensates for $c_t$. Further, we note that each candidate $c_{t}
  \in W'$ is assigned to a different ``compensating'' candidate (see
  Figure~\ref{fig:shifting2} and consider how the assignment is
  defined, starting from the highest values of~$t$ and decreasing $t$
  one by one).  We conclude that the score of $W'$ increases in $E'$
  by a value that is not greater than the increase of the score of
  $W$. Since $W'$ was chosen arbitrarily, we get that $W$ is a winner
  in $E'$, which completes the proof.
\end{proof}

As far as applications of multiwinner voting goes, prefix monotonicity
does not seem to have as clear-cut interpretation as non-crossing
monotonicity. Nonetheless, in the next section we will see how its
relaxed variant is useful in characterizing representation-focused
rules (and how this characterization can be interpreted in the context
of diversity-oriented committee elections).

\subsubsection{Top-Member Monotonicity and Representation-Focused  Rules}
\label{sec:top-member_monotonicity}

Our goal in this section is to provide an axiomatic characterization
of representation-focused rules. The first tool that we employ for
this task is $1$-prefix monotonicity (recall
Definition~\ref{def:prefix_monotonicity}), which we rename as
\emph{top-member monotonicity}. Intuitively, top-member monotonicity
requires that if in some vote $v$ we shift forward the highest-ranked
member of a given winning committee, then this committee remains
as a winning one.
Since top-member monotonicity is a relaxed variant of
non-crossing monotonicity (and of prefix monotonicity), it is satisfied
by all weakly separable rules and alone is insufficient to
characterize representation-focused rules.
Thus we will also use the notion of narrow-top consistency, defined
below (which, in fact, is a relaxed form of the solid coalitions
property of Elkind et
al.~\cite{elk-fal-sko-sli:j:multiwinner-properties}, itself motivated
by a much stronger notion of Dummet~\cite{dum:b:voting}).

\begin{definition} %
  A multiwinner rule $\calR$ satisfies \emph{narrow-top consistency}
  if for each election $E = (C,V)$ and each $k \in [|C|]$ the
  following holds: If there exists a set of at most $k$ candidates
  $S$, such that each voter in $V$ ranks some candidate from $S$
  first and each member of $S$ is ranked first by some voter,
  then for each $W \in \calR(E,k)$ it holds that $S \subseteq
  W$.
\end{definition}

Together, top-member monotonicity and narrow-top consistency exactly
characterize the class of representation-focused rules (within the
class of committee scoring rules). We prove this result formally
below, and then we explain the roles of both our axioms intuitively.

\begin{theorem}\label{thm:representation_focused}
  Let $\calR_f$ be a committee scoring rule.
  $\calR_f$ is representation-focused if and only if it 
  satisfies top-member monotonicity and narrow-top consistency. 
\end{theorem}

\begin{proof}
  It is apparent that each representation-focused rule satisfies both
  top-member monotonicity and narrow-top consistency.  Suppose that
  $\calR_f$ is a committee scoring rule, defined through a family
  $f=(f_{m,k})_{k\le m}$ of scoring functions $f_{m,k}\colon [m]_k\to
  \reals_{+}$, that satisfies these two properties. We will show that
  $\calR_f$ is representation-focused.

  Let us fix the number of candidates $m$ and the committee size
  $k$. Since $\calR_f$ satisfies $1$-prefix monotonicity (top-member
  monotonicity), by Lemma~\ref{lem:t_prefix_monotonicity} we have that
  there exists a function $h$ such that for each $p\in [m]$, each $U
  \in P_{m, k}(1, p)$ and the committee position $U'$, obtained from
  $U$ by replacing position $p$ with $p-1$, we have $h(p-1) = f(U') -
  f(U)\ge 0$.

  Let $I = (i_1, \ldots, i_k)$ and $J = (j_1, \ldots, j_k)$ be such
  that $i_1 = j_1$. We will show that $f_{m,k}(I) = f_{m,k}(J)$, which is
  sufficient to prove that $\calR_f$ is representation focused. For
  the sake of contradiction, let us assume that this is not the case,
  and without loss of generality, let us assume that $f_{m,k}(I) >
  f_{m,k}(J)$. There exists a positive integer $\eta$ such that $\eta f_{m,k}(I) >
  \eta f_{m,k}(J) + k f_{m,k}(1, \ldots, k)$.

  Let us fix a vote $v$ and let $W$ and $W'$ denote the committees
  that stand in $v$ on positions $I$ and $J$, respectively. Note that
  they have a common member $d$ who stands on position $i_1 = j_1$ and
  is highest ranked by $v$ in both committees. Consider an election
  $E$ with $\eta$ copies of vote $v$ and with $k$ votes such that for
  each candidate $c \in W'$ there is one vote who ranks $c$ first and
  the remaining candidates in some fixed, arbitrary way. In this
  election the score of $W$ is at least equal to $\eta f(I)$ and the
  score of $W'$ is at most equal to $\eta f(J) + k f(1, \ldots k)$.
  Thus the score of $W$ is higher than the score of $W'$. If $i_1 =
  j_1=1$ we get a contradiction immediately since by the narrow-top
  consistency $W'$ must be winning.

  If $i_1 = j_1\ne 1$, we construct election $E'$ by shifting, in each
  copy of $v$, the candidate $d\in W\cap W'$ to the top position. In
  comparison to $E$, the scores of committees $W$ and $W'$ in $E'$
  increase by the same value $\eta\big(h(1) - h(i_1)\big)$. As a
  result, $W$ has a higher score than $W'$ also in $E'$. This,
  however, contradicts narrow-top consistency, since all top positions
  in this profile are occupied by candidates from $W'$. This proves
  that $f(I) = f(J)$, and completes the reasoning.
\end{proof}

Let us now explain intuitively the interplay between top-member monotonicity and
narrow-top consistency in the characterization of representation-focused
rules. If we applied similar reasoning as we used in the
proofs of Theorems~\ref{thm:weaklysep}
and~\ref{thm:prefixMonAndDecomposable} to top-member monotonicity,
then we could show that if a committee scoring rule $\calR_f$ is
top-member monotone then its scoring functions are of the form:
\begin{equation}\label{eq:tmm}
  f_{m,k}(i_1, \ldots, i_k) = \gamma_{m,k}(i_1) + g_{m,k-1}(i_2, \ldots, i_k) \textrm{,}
\end{equation}
where $\gamma = (\gamma_{m,k})_{k \leq m}$ is a family of
single-winner scoring functions and $g = (g_{m,k-1})_{k-1 \le m}$ is a
family of committee scoring functions. Requiring that $\calR_f$ is
also narrow-top consistent ensures that the functions $g_{m,k}$ are,
in fact, constant, and in consequence gives that $\calR_f$ is
representation-focused. Since all decomposable committee scoring rules
are already of the form presented in equation~\eqref{eq:tmm} (and, in
fact, they are of a far more restricted form), we have the following
corollary (the proof follows directly from the preceding reasoning,
but straightforward calculations also show it directly; we omit these details).
\begin{corollary}
  If a decomposable committee scoring rule is narrow-top consistent
  then it is representation focused.
\end{corollary}

Representation-focused rules generally, and the Chamberlin Courant
rule specifically, are often considered in the context of selecting
diverse
committees~\cite{elk-fal-sko-sli:j:multiwinner-properties,fal-sko-sli-tal:b:multiwinner-trends}. While
there is no clear definition of what a ``diverse committee'' is,
researchers often use this term intuitively, to mean that as many
voters as possible can find a committee member that they rank highly
(if a voter $v$ ranks some committee member $c$ highly, then we could
say that $c$ ``covers'' the views of $v$, so some authors speak of
``diversity/coverage''; see the works of Ratliff and
Saari~\cite{rat-sar:j:diverse-committees}, Bredereck et
al.~\cite{bre-fal-iga-lac-sko:c:diverse-committees}, Celis et
al.~\cite{cel-hua-vis:diverse-committees}, and Izsak et
al.~\cite{interclass} for a different view regarding diverse
committees).
Theorem~\ref{thm:representation_focused} justifies the use of
representation-focused rules to seek committees that are diverse in
this sense. Indeed, if there is a committee such that every voter
ranks one of its members on top, then certainly this committee
``covers'' the ``diverse'' views of all the voters; narrow-top
consistency ensures that this committee is selected. On the other
hand, if there is a committee $W$ and we agree that it ``covers'' the
views of sufficiently many voters, then if some voter ranks his or her
highest-ranked committee member even higher (i.e., this voter realizes
that the candidate represents his or her views even better), then
certainly we should still view $W$ as ``covering'' the views of
sufficiently many voters; this is ensured by top-member monotonicity.

\subsection{Committee Enlargement Monotonicity and Separable Rules}
\label{sec:com_mon_sep_rules}

In this section we consider the committee enlargement monotonicity
axiom. While it is markedly different from the notions that we used in
the previous sections, it still has a clear monotonicity flavor:
Informally speaking, it requires that if $W$ is a size-$k$ winning
committee for some election, then there also is a size-$(k+1)$ winning
committee for this election that includes all the members of $W$ (the
actual definition is more complicated due to possible ties; its exact
form is due to Elkind et
al.~\cite{elk-fal-sko-sli:j:multiwinner-properties}, but it was
already studied by Barber\`a
and~Coelho~\cite{bar-coe:j:non-controversial-k-names} for resolute
multiwinner rules, and in the literature on apportionment rules it is
well known as \emph{house monotonicity}~\cite{Puke14a,
  bal-you:b:polsci:representation}).

\begin{definition}[Elkind et al.~\cite{elk-fal-sko-sli:j:multiwinner-properties}]
  A multiwinner election rule $\calR$ satisfies \emph{committee monotonicity} if for each $m$ and $k$, $1 \leq k < m$, and for each election $E$ the following two conditions hold: 
  \begin{itemize}
  \item[(1)] for each $W \in \calR(E,k)$ there exists $W' \in \calR(E,k+1)$ such that $W \subseteq W'$; 
  \item[(2)] for each $W \in \calR(E,k+1)$ there exists $W' \in \calR(E,k)$ such that $W' \subseteq W$.
  \end{itemize}
\end{definition} 

This section is almost completely dedicated to showing that in the
class of committee scoring rules, committee enlargement monotonicity
characterizes exactly the class of separable rules.

\begin{theorem}\label{thm:committee_monot_and_sep}
  Let $\calR_f$ be a committee scoring rule.
  $\calR_f$ is committee-enlargement monotone if and only if $\calR_f$ is
  separable.

\end{theorem}
\noindent
Before we provide the proof of
Theorem~\ref{thm:committee_monot_and_sep}, we first introduce useful
notation and tools. Given two elections $E_1 = (C,V_1)$ and $E_2 =
(C,V_2)$, by $E_1 + E_2$ we mean election $(C, V_1+V_2)$, whose voter
collection is obtained by concatenating the voter collections of $E_1$
and $E_2$. For an election $E = (C,V)$ and a positive integer
$\lambda$, by $\lambda E$ we mean election $(C, \lambda V)$, whose
voter collection consists of $\lambda$ concatenated copies of $V$.
We will heavily rely on the
properties of the following elections
(let $C$ be some set of $m$ candidates; this set will always be clear
from the context when we use the notation introduced below):
\begin{enumerate}
\item For each candidate $c \in C$, by $\zeta(c)$ we denote the
  election with $(m-1)!$ voters who all rank $c$ as their most
  preferred candidate, followed by each possible permutation of the
  remaining $m-1$ candidates.

\item For each subset $S \subseteq C$, we define election $\zeta(S)$
  to be $\sum_{c \in S} \zeta(c)$ (i.e., it is a concatenation of the
  elections $\zeta(c)$ for each $c \in S$).

\end{enumerate}
The next two lemmas describe which committees win in elections
$\zeta(c)$ and $\zeta(S)$.

\begin{lemma}\label{lemma:profile_zeta1}
  Fix $m$ and $k$, and consider a non-degenerate committee scoring rule
  $\calR$ defined through a scoring function $f_{m,k}$.  The set of
  winners for $\zeta(c)$ consists of all committees that contain $c$.
\end{lemma}
\begin{proof}
  Since $\calR$ is non-degenerate, there exists $i$ such that
  $f_{m,k}(i+1, \ldots, i+k) > f_{m,k}(i+2, \ldots, i+k+1)$. By the
  fact that election $\zeta(c)$ is symmetric with respect to all the
  candidates except $c$, we see that all committees that contain $c$
  have the same $f_{m,k}$-score. Similarly, all committees that do not
  contain $c$ also have the same score.  Consider a committee $W$ such
  that $c \notin W$. Let $c'$ be an arbitrary member of $W$ and let
  $W' = (W \setminus \{c'\}) \cup \{c\}$. Naturally, in each vote the
  position of committee $W'$ dominates that of $W$. Further, there
  exists a vote where $W$ has committee position $(i+2, \ldots,
  i+k+1)$, and $W'$ has position $(1, i+2, \ldots, i+k)$. From this
  vote $W$ gets score $f_{m,k}(i+2, \ldots, i+k+1)$ and $W'$ gets
  score $f_{m,k}(1, i+2, \ldots, i+k) > f_{m,k}(i+1, \ldots,
  i+k)$. Thus the score of $W'$ in $\zeta(c)$ is higher than that of
  $W$. This completes the proof.
\end{proof}

\begin{lemma}\label{lemma:profile_zeta2}
  Fix $m$, $k$, and $S \subseteq C$, and consider a non-degenerate
  committee scoring rule $\calR$ defined through a scoring function
  $f_{m,k}$. If $|S| \geq k$ then the set of winning committees of
  $\zeta(S)$ consists of all the committees $W$ such that $W \subseteq
  S$.  Otherwise, it consists of all the committees $W$ such that $S
  \subseteq W$.
\end{lemma}
\begin{proof}
  Consider election $\zeta(c)$ and let $x$ and $y$ denote the scores
  of committees, respectively, containing $c$ and not containing
  $c$. From Lemma~\ref{lemma:profile_zeta1} it follows that $x >
  y$. Consider the case when $|S| \geq k$ (the proof for the other
  case follows by analogous reasoning). The score of a committee $W$
  such that $W \subseteq S$ is equal to $kx + (|S| - k)y$. For each
  committee $W'$ with $W' \not\subseteq S$, its score is at most equal
  to $(k-1)x + (|S| - k+1)y < kx + (|S| - k)y$.
\end{proof}

In the following observation we analyze the scores of candidates and
committees in the elections we will be using in the proof of
Theorem~\ref{thm:committee_monot_and_sep}.

\begin{observation}\label{obs:zeta_construction2}
  Consider two committees, $W_1$ and $W_2$, with $W_1 \setminus W_2 =
  \{c_1\}$ and $W_2 \setminus W_1 = \{c_2\}$.  By symmetry of our
  construction, for each single-winner scoring function $f_{m,1}$, the
  $f_{m,1}$-scores of the candidates $c_1$ and $c_2$ 
  are the same in election $\zeta(W_1 \cup W_2)$,
  are the same in election $\zeta(W_1 \cap W_2)$,
  and are the same in election $\zeta\big(\{c_1,
  c_2\}\big)$.
  Further, in each of these three elections, the $f_{m,1}$-scores of
  any two candidates $c, c' \in W_1 \cap W_2$ are equal. If $f_{m,1}$
  is nontrivial, then in $\zeta(W_1 \cup W_2)$, $\zeta(W_1 \cap W_2)$
  and $\zeta\big(\{c_1, c_2\}\big)$ the $f_{m,1}$-scores of candidates
  $c_1$ and $c_2$ are, respectively, the same, lower, and higher than
  the $f_{m,1}$-score of any other candidate $c\in W_1\cup W_2$.
  Also, for each committee scoring function $f_{m,k}$, the
  $f_{m,k}$-scores of committees $W_1$ and $W_2$ are the same in
  $\zeta(W_1 \cup W_2)$, are the same in $\zeta(W_1 \cap W_2)$, and
  are the same in $\zeta\big(\{c_1, c_2\}\big)$.  In $\zeta\big(W_1
  \cap W_2 \setminus \{c\}\big)$, where $c\neq c_1, c_2$, the $f_{m,k}$-scores
  of $W_1$ and $W_2$ are equal, and the $f_{m,1}$-score of $c$ is lower than the
  $f_{m,1}$-score of any other candidate from $W_1 \cap W_2$.
\end{observation}

In the next lemma we handle the possibility that the rule $\calR_f$ in
Theorem~\ref{thm:committee_monot_and_sep} may be trivial.
(Note that
it is not the case that if a committee-enlargement monotone
multiwinner rule always outputs all size-$1$ committees then it
also always outputs all size-$k$ committees for larger values of~$k$.
The result below excludes this behavior for the subclass of committee scoring rules.)

\begin{lemma}
  \label{triviality}
  Suppose that $\calR_f$ is a committee scoring rule defined by a
  family $f=(f_{m,k})_{k\le m}$ of scoring functions, such that
  $\calR_f$ is committee-enlargement monotone and $f_{m,1}$ is
  constant. Then $f_{m,k}$ is constant for every $k\le m$.
\end{lemma}

\begin{proof}
  $\calR_f$ is trivial for $k=1$ and we will show that, in fact, it is
  trivial for all $k$. The proof follows by induction. Let us assume
  that $\calR_f$ is trivial for some $k = p - 1$, i.e., that
  $f_{m,p-1}$ is constant. For the sake of contradiction let us assume
  that $f_{m,p}$ is not trivial, hence $f_{m,p}(1, \ldots, p) > f_{m,
    p}(m-p+1, \ldots m)$. Let $i$ be the smallest positive integer
  such that $f_{m,p}(i+1, \ldots, i + p) > f_{m, p}(i+2, \ldots
  i+p+1)$. Consider an election where a certain candidate $c$ is
  always in position $i + p + 1$, the positions $i+p+2, \ldots, m$ are
  also always occupied by the same candidates, and on positions $1,
  \ldots, i + p$ there are always the same candidates, call the set of
  these candidates $S$, but in all possible permutations. We can see
  that the $f_{m,p}$-scores of committees that consists only of
  candidates from $S$ are higher than the $f_{m,p}$-scores of
  committees that contain $c$ (the reasoning is very similar to the
  one given in the proof of Lemma~\ref{lemma:profile_zeta1}). This,
  however, contradicts committee monotonicity, since by our inductive
  assumption, for $k = p-1$ all committees were winning, and so for
  $k=p$ there should be at least one winning committee containing $c$.
\end{proof}

We are nearly ready to present the proof of
Theorem~\ref{thm:committee_monot_and_sep}. The final piece of notation
that we will need is as follows. Given two committee positions $I =
(i_1, \ldots, i_k)$ and $J = (j_1, \ldots, j_k)$, we will sometimes
treat them as sets rather than sequences. For example, by $|I \cap J|$
we will mean the number of single-candidate positions that occur
within both $I$ and $J$, and we will say that $i \in I$ if there is
some $t$ such that $i = i_t$.

\newcommand{\singlescore}{f_{m,1}\hbox{-}\score}
\newcommand{\pscore}{f_{m,p}\hbox{-}\score}
\newcommand{\plonescore}{f_{m,p-1}\hbox{-}\score}

\begin{proof}[Proof of Theorem~\ref{thm:committee_monot_and_sep}] 
  Each separable committee scoring rule is committee-enlargement
  monotone and we focus on proving the converse.

  Let $\calR_f$ be the committee-enlargement monotone committee
  scoring rule defined through a family $f=(f_{m,k})_{k\le m}$ of
  scoring functions.  Let us fix the number of candidates in our
  elections to be $m$. $\calR_f$ assigns a score to each committee of
  each size and, in particular, for $k = 1$, given an election
  $E=(C,V)$ it assigns $f_{m,1}$-score to each candidate (singleton):
  \begin{align*}
    \singlescore_E(c) = \sum_{v_i \in V}f_{m,1}(\pos_{v_i}(c)) \text{,}
  \end{align*}
  We will show by induction on $k$ that for each election $E$, a
  size-$k$ committee $W$ is winning under $\calR_f$ if and only if it
  consists of candidates with the $k$ highest $\singlescore\mathrm{s}$.

  The base for the induction, for $k=1$, follows immediately from
  the definition of $\calR_f$.  Now, to prove the inductive step, let us
  assume that for each $k < p$ and for each election $E$ it holds that
  $W \in \calR_f(E, k)$ if and only if it consists of $k$ candidates
  with the highest $\singlescore\mathrm{s}$.
  We will show that this is also the case for $k = p$.  By
  Lemma~\ref{triviality} we may assume that $\calR_f$ for $k=1$ is
  nontrivial.\medskip

  Our first task is to show that  whenever:
  \begin{equation}
    \label{ind_hyp} 
    f_{m,p}(1, \ldots, p) = f_{m,p}(i+1, \ldots, i+p)
  \end{equation}
  for some $i\in [m]$, then $f_{m,1}(1) = \cdots = f_{m,1}(i+p)$. For
  the sake of contradiction let us assume that $f_{m,p}(1, \ldots, p) = f_{m,p}(i+1,
  \ldots, i+p)$ and $f_{m,1}(1) > f_{m,1}(i+p)$. We first show that it
  must hold that $f_{m,1}(1) = \cdots = f_{m,1}(i+p-1)$. To see why
  this is the case, consider an election with a single vote $c_1 \succ
  c_2 \succ \cdots \succ c_m$. In such an election, committee $\{c_1,
  \ldots, c_p\}$ always wins and, since $f_{m,p}(1, \ldots, p) =
  f_{m,p}(i+1, \ldots, i+p)$, we have that committee $W_i = \{c_{i+1}, \ldots, c_{i+p}\}$ also wins. By
  committee-enlargement monotonicity we know that some
  size-$(p-1)$ subcommittee of $W_i$ wins for committee size $p-1$
  and, in particular, by weak dominance we get that
  certainly $W'_i = \{c_{i+1}, \ldots, c_{i+p-1}\}$ wins. Thus,
  by the inductive hypothesis it must be the case that:
  \begin{align*}
    f_{m,1}(1) + \ldots + f_{m,1}(p-1) = f_{m,1}(i+1) + \ldots + f_{m,1}(i+p-1)
  \end{align*} 
  which implies that $f_{m,1}(1) = f_{m,1}(i+p-1)$ and, thus, that
  $f_{m,1}(1) = f_{m,1}(2) = \cdots = f_{m,1}(i + p - 1)$. Since we assumed that $f_{m,1}(1) > f_{m,1}(i+p)$,
  it must be the case that $f_{m,1}(i+p-1) > f_{m,1}(i+p)$.

  Now we show that the assumption that $f_{m,1}(i+p-1) > f_{m,1}(i+p)$ also leads
  to a contradiction.
  Consider an election $E$ with two votes:
  \begin{align*}
    v_1\colon& c_1 \succ c_2 \succ \ldots \succ c_{i+p-1} \succ c_{i+p} \succ \ldots \succ c_m \text{,} \\
    v_2\colon& c_1 \succ c_2 \succ \ldots \succ c_{i+p} \succ c_{i+p-1} \succ \ldots \succ c_m \text{,}
  \end{align*}
  which differ only in the order of $c_{i+p-1}$ and $c_{i+p}$.  For
  each $j < i+p-1$, we have $\singlescore_E(c_j)=2f_{m,1}(j)$ and this
  value is higher than the $f_{m,1}$-scores of $c_{i+p-1}$ and
  $c_{i+p}$. By the inductive hypothesis, this means that for $k =
  p-1$ there is no winning committee that contains either $c_{i+p-1}$
  or $c_{i+p}$. Thus, by committee-enlargement monotonicity, we infer
  that no winning committee for $k = p$ contains both $c_{i+p-1}$ and
  $c_{i+p}$. On the other hand, for $k = p$ due to \eqref{ind_hyp} the
  $f_{m,p}$-score of committee $\{c_{i+1}, \ldots, c_{i+p-1}, c_{i+p}\}$ is
  the highest among committees of size $p$, which gives a
  contradiction.\medskip

  Next, let $i$ be the smallest value such that $f_{m,p}(i + 1,
  \ldots, i + p) > f_{m,p}(i+2, \ldots, i+p+1)$. We will show that
  $f_{m,1}(i+p) > f_{m,1}(i+p+1)$. Again, for the sake of
  contradiction, let us assume that this is not the case and $f_{m,1}(i+p) = f_{m,1}(i+p+1)$. By our
  previous reasoning we have that $f_{m,1}(1) = \cdots =
  f_{m,1}(i+p)$. Consider an election where a fixed candidate $c$
  stands on position $i+p+1$ and some set of $i+p$ candidates stands
  on the first $i+p$ positions in all possible permutations. In such
  an election there is no winning committee of size $p$ that contains
  $c$. However, by the inductive hypothesis, a winning committee of
  size $p-1$ containing $c$ does exist. This contradicts
  committee-enlargement monotonicity.\medskip

  By the above reasoning, we can find two committee positions $I^*$
  and $J^*$, for committees of size $p$, such that $|I^* \cap J^*| =
  p-1$, $f_{m,p}(I^*) > f_{m,p}(J^*)$, and $\sum_{i \in I^*}f_{m,1}(i)
  > \sum_{i \in J^*}f_{m,1}(i)$.  Let us arrange all committee
  positions from $[m]_k$ in a sequence $\calS$ so that for each two
  consecutive elements $I$ and $J$ in $\calS$ it holds that $|I \cap
  J| = p-1$. This is possible (see the construction based on Johnson
  graphs in Lemma 8 of the work of
  Skowron~et~al.~\cite{sko-fal-sli:t:axiomatic-committee}).  We claim
  that for each two consecutive elements of sequence $\calS$, call
  them $I$ and $J$, it holds that:
  \begin{align}\label{eq:equalDerivatives}
    \frac{f_{m,p}(I) - f_{m,p}(J)}{f_{m,p}(I^*) - f_{m,p}(J^*)} =
    \frac{\sum_{i \in I}f_{m,1}(i) - \sum_{i \in J}f_{m,1}(i)}{\sum_{i
        \in I^*}f_{m,1}(i) - \sum_{i \in J^*}f_{m,1}(i)} \text{.}
  \end{align}
  (Note that the above expression is well defined. There is no
  division by zero because we selected $I^*$ and $J^*$ so that
  $f_{m,p}(I^*) \neq f_{m,p}(J^*)$ and $\sum_{i \in I^*}f_{m,1}(i)
  \neq \sum_{i \in J^*}f_{m,1}(i)$.)  For the sake of contradiction,
  let us assume that equality~\eqref{eq:equalDerivatives} does not
  hold for some $I$ and $J$, and let us assume that
  there exist $x, y \in \naturals$ such that:
  \begin{align}\label{eq:nonequalDerivatives}
    \frac{f_{m,p}(I) - f_{m,p}(J)}{f_{m,p}(I^*) - f_{m,p}(J^*)} >
    \frac{x}{y} > \frac{\sum_{i \in I}f_{m,1}(i) - \sum_{i \in
        J}f_{m,1}(i)}{\sum_{i \in I^*}f_{m,1}(i) - \sum_{i \in
        J^*}f_{m,1}(i)}
  \end{align}
  Let $W_1$ and $W_2$ be two fixed committees with $|W_1 \cap W_2| =
  p-1$. Let $W_1 \setminus W_2 = \{c_1\}$ and $W_2 \setminus W_1 =
  \{c_2\}$. We construct election $Q$ in which there are $x$ votes
  where $W_1$ stands on position $I^*$ and $W_2$ on position $J^*$,
  and $y$ votes where $W_1$ stands on position $J$ and $W_2$ on
  position $I$. In $Q$ the score of $W_2$ is equal to $xf_{m,p}(J^*) +
  yf_{m,p}(I)$, and the score of $W_1$ is equal to $xf_{m,p}(I^*) +
  yf_{m,p}(J)$. By inequality~\eqref{eq:nonequalDerivatives}, we see
  that the $f_{m,p}$-score of $W_2$ in $Q$ is greater than the
  $f_{m,p}$-score of $W_1$, yet the sum of the $f_{m,1}$-scores of
  members of $W_2$ is lower than that of the members of $W_1$, which
  means that the $f_{m,1}$-score of $c_1$ in $Q$ is greater than the
  $f_{m,1}$-score of~$c_2$ (these are the only candidates in which the
  two committees differ).  If
  inequality~\eqref{eq:nonequalDerivatives} were reversed (i.e., if we
  replaced both occurrences of ``$>$'' with ``$<$'') then the same
  construction would still work but we would have to reverse the roles
  of $W_1$ and $W_2$ and of $c_1$ and $c_2$.

  We construct election $Q_s$ by taking each possible permutation
  $\sigma$ of the candidates from $W_1 \cap W_2$ and by concatenating
  all elections of the form $\sigma(Q)$ (where $\sigma(Q)$ is an
  election that results from applying $\sigma$ to the candidates in
  all the preference orders within $Q$).  Thus, intuitively, $Q_s$ can
  be viewed as a symmetric version of $Q$, where symmetry is with
  respect to the candidates in $W_1 \cap W_2$.  In particular in $Q_s$
  it holds that:
  \begin{enumerate}[(a)]
  \item the $f_{m,p}$-score of $W_2$ is higher than the $f_{m,p}$-score of $W_1$,
  \item the $f_{m,1}$-score of $c_1$ is higher than that of $c_2$, and
  \item the $f_{m,1}$-scores of all candidates from $W_1 \cap W_2$ are equal.
  \end{enumerate}

\begin{figure}[t!]
  \begin{center}
    \includegraphics[scale=0.55]{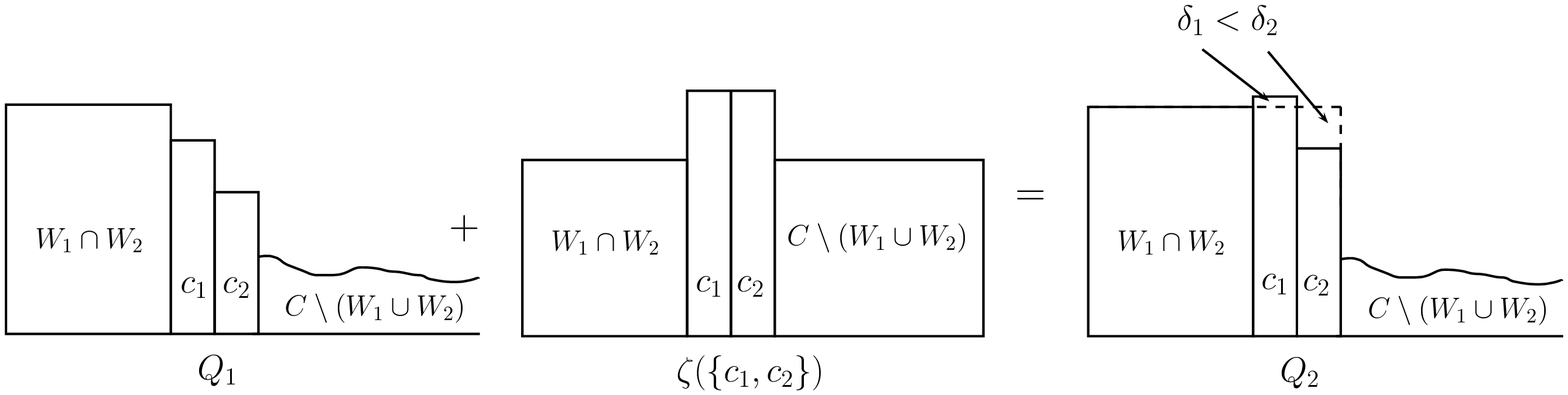}
  \end{center}
  \caption{The construction of election $Q_2$ from $Q_1$ in case
    (i). The ``$x$-axis'' corresponds to candidates and the
    ``$y$-axis'' corresponds to their $f_{m,1}$-scores.  Here,
    $\delta_1$ and $\delta_2$ denote, respectively, the differences
    between the scores of $c$ and $c_1$ and the difference between the
    scores of $c$ and $c_2$. The shape of the election
    $\zeta\big(\{c_1, c_2\}\big)$  is
    justified in Observation~\ref{obs:zeta_construction2}.}
  \label{fig:proof_transformation1}
  \begin{center}
    \includegraphics[scale=0.55]{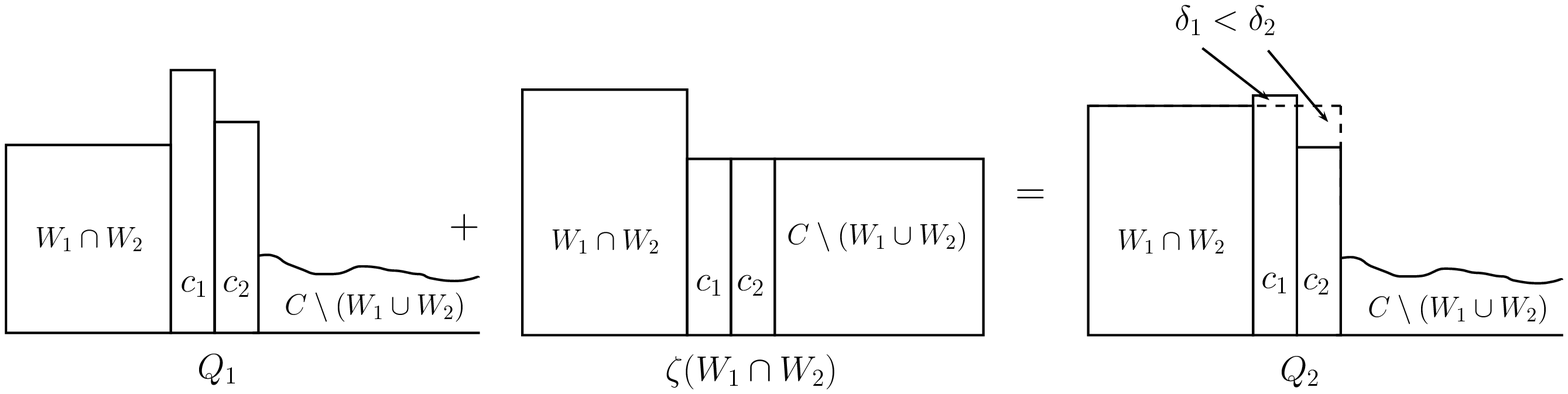}
  \end{center}
  \caption{The construction of election $Q_2$ from $Q_1$ in case (ii);
    the interpretation of the figure is the same as for
    Figure~\ref{fig:proof_transformation1}. The shape of the election
    $\zeta(W_1 \cap W_2)$ is justified in
    Observation~\ref{fig:proof_transformation1}.}
  \label{fig:proof_transformation2}
  \end{figure}

  There exists $\lambda \in \naturals$ such that in election $Q_1 =
  \lambda \zeta(W_1 \cup W_2) + Q_s$ each candidate from $W_1 \cup
  W_2$ has higher $f_{m,1}$-score than each candidate outside of $W_1
  \cup W_2$. By Observation~\ref{obs:zeta_construction2}, it is clear
  that the $f_{m,1}$-score of candidate $c_1$ in $Q_1$ is higher than
  that of candidate $c_2$. Intuitively, this transformation allows us
  to focus only on the candidates from $W_1 \cup W_2$.

  Now, let $c$ be a fixed arbitrary candidate from $W_1 \cap W_2$.  We
  construct election $Q_2$ using $Q_1$ in the following way.
  \begin{enumerate}[(i)]
  \item If in $Q_1$ the $f_{m,1}$-score of $c$ is higher than the
    $f_{m,1}$-score of $c_1$, then we define $Q_2$ as a linear
    combination $Q_2 = \lambda_1 Q_1 + \lambda_2 \zeta\big(\{c_1,
    c_2\}\big)$ (this is depicted in
    Figure~\ref{fig:proof_transformation1}).
  \item Otherwise, i.e., if in $Q_1$ the $f_{m,1}$-score of $c_1$ is at least as
    high as the $f_{m,1}$-score of $c$, then we define $Q_2$ as a linear
    combination $Q_2 = \lambda_1 Q_1 + \lambda_2 \zeta(W_1 \cap W_2)$
    (this is depicted in Figure~\ref{fig:proof_transformation2}).
  \end{enumerate}
  In each of these two cases we choose the coefficients $\lambda_1$
  and $\lambda_2$ so that in $Q_2$ it holds that the $f_{m,1}$-score of $c_1$
  is higher than that of $c$, which is higher than the $f_{m,1}$-score of $c_2$.
  Further, we choose $\lambda_1$ and $\lambda_2$ so that the
  difference between the $f_{m,1}$-scores of $c$ and $c_1$ is smaller than the
  difference between the $f_{m,1}$-scores of $c$ and $c_2$. Formally:
  \begin{align}\label{eq:proper_differences}
    \singlescore_{Q_2}(c_1) - \singlescore_{Q_2}(c) < \singlescore_{Q_2}(c) - \singlescore_{Q_2}(c_2) \text{.}
  \end{align}
  Why is it possible to choose such $\lambda_1$ and $\lambda_2$? We
  will give a formal argument for Case (i) and it will be clear that
  this reasoning can be repeated for Case (ii). Let:
  \begin{align*}
    \Delta_1 = \singlescore_{Q_1}(c) - \singlescore_{Q_1}(c_1) \quad \text{and}
    \quad \Delta_2 = \singlescore_{Q_1}(c_1) - \singlescore_{Q_1}(c_2).
  \end{align*}
  Further, let $\Delta_3$ denote the difference between the
  $f_{m,1}$-scores of the candidates from $\{c_1, c_2\}$ and the
  $f_{m,1}$-scores of the candidates outside of $\{c_1, c_2\}$ in
  $\zeta\big(\{c_1, c_2\}\big)$. Naturally, there exist natural
  numbers $p, q \in \naturals$ such that:
  \begin{align*}
    \Delta_1 < \frac{p}{q} \Delta_3 < \Delta_1 + \frac{1}{2} \Delta_2
    \text{.}
  \end{align*}
  We set $\lambda_1 = q$ and $\lambda_2 = p$, and from the above
  inequality we get that:
  \begin{align}\label{eq:lambda1lambda2set}
    \lambda_1 \Delta_1 < \lambda_2 \Delta_3 < \lambda_1\left(\Delta_1 +
    \frac{1}{2} \Delta_2\right) \text{.}
  \end{align}
  Observe that:
  \begin{align*}
    \singlescore_{Q_2}(c_1) - \singlescore_{Q_2}(c) &= - \lambda_1 \Delta_1 + \lambda_2 \Delta_3 > 0 \text{,} \\
    \singlescore_{Q_2}(c) - \singlescore_{Q_2}(c_2) &= \lambda_1 (\Delta_1 + \Delta_2) - \lambda_2 \Delta_3 > - \lambda_1 \Delta_1 + \lambda_2 \Delta_3 \\
    &= \singlescore_{Q_2}(c_1) - \singlescore_{Q_2}(c) \text{.}
  \end{align*}
  (The second inequality above is equivalent to $2\lambda_1\Delta_1 +
  \lambda_1\Delta_2 - 2\lambda_2 \Delta_3 > 0$, and thus follows from
  inequality~\eqref{eq:lambda1lambda2set}.)  
  \begin{figure}[t!]
    \begin{center}
      \includegraphics[scale=0.55]{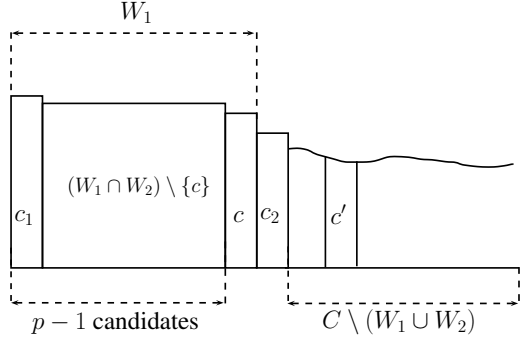}
    \end{center}
    \caption{Illustration of election $Q_3$. The interpretation of the
      figure is the same as for
      Figure~\ref{fig:proof_transformation1}.}
    \label{fig:proof_transformation3}
  \end{figure}
  Next, we construct $Q_3$ as $Q_3 = \lambda_4 Q_2 + \zeta\big(W_1
  \cap W_2 \setminus \{c\}\big)$, where $\lambda_4$ is a very large
  number so that in $Q_3$ we still have that $\singlescore_{Q_3}(c_1)
  > \singlescore_{Q_3}(c) > \singlescore_{Q_3}(c_2)$ and that in $Q_3$
  inequality~\eqref{eq:proper_differences} still holds, yet the
  $f_{m,1}$-score of $c$ is slightly lower than the $f_{m,1}$-scores
  of the other candidates from $W_1 \cap W_2$. Election $Q_3$ is
  depicted in Figure~\ref{fig:proof_transformation3}.

  Given the $f_{m,1}$-scores of the candidates in $Q_3$, by our
  inductive assumption the unique size-$(p-1)$ winning committee
  consists of $c_1$ and all the candidates from $W_1 \cap W_2
  \setminus \{c\}$. By committee-enlargement monotonicity, we conclude
  that all size-$p$ winning committees for $Q_3$ are of the form
  $\{c_1\} \cup (W_1 \cap W_2 \setminus \{c\}) \cup \{c'\}$, where
  $c'$ is some other candidate. Let $W'$ be one such winning
  committee. We know that it cannot be the case that $c' = c$
  (i.e., $W_1$ cannot be winning in $Q_3$). This is so, because in
  $Q_3$ the $f_{m,p}$-score of $W_2$ is higher than the
  $f_{m,p}$-score of $W_1$ (since it was higher already in $Q_s$, and
  we added only elections which are symmetric with respect to $W_1$
  and $W_2$---this symmetry follows from
  Observation~\ref{obs:zeta_construction2}).  Thus, by the properties
  of the $f_{m,1}$-scores of the candidates (see
  Figure~\ref{fig:proof_transformation3}), $c'$ must be some candidate
  such that:
  \begin{align*}
    \singlescore_{Q_3}(c') \leq \singlescore_{Q_3}(c_2) \text{,}
  \end{align*}
  and, in particular $c'$ may simply be $c_2$ (but it also may be some
  other candidate).  From the above inequality and from
  inequality~\eqref{eq:proper_differences} (which holds for $Q_3$ as
  well) we get that:
  \begin{align*}
    \singlescore_{Q_3}(c_1) - \singlescore_{Q_3}(c) < \singlescore_{Q_3}(c) -    \singlescore_{Q_3}(c') \text{.}
  \end{align*}
  We note that in $Q_3$ committee $W'$ has higher score than any
  committee containing $c$ (otherwise, since $W'$ is winning in $Q_3$
  and by the above analysis, it would mean that $W_1$ is winning in
  $Q_3$, which is not the case).  

  Next, we construct election $Q_3'$ by swapping candidates $c_1$ and
  $c'$ in each vote in $Q_3$. Committee $W'$ is also winning in $Q_3'$
  and thus it has higher score in $Q_3'$ than any committee containing
  $c$. Similarly, by symmetry, we infer that:
  \begin{align*}
    \singlescore_{Q_3'}(c') - \singlescore_{Q_3'}(c) < \singlescore_{Q_3'}(c) - \singlescore_{Q_3'}(c_1) \text{.}
  \end{align*}
  Finally, we construct election $Q_4$ by taking one copy of $Q_3$ and
  one copy of $Q_3'$. Observe that:
  \begin{align*}
    \singlescore_{Q_4}(c) &= \singlescore_{Q_3}(c) + \singlescore_{Q_3'}(c) \\
    &> \singlescore_{Q_3}(c_1) - \singlescore_{Q_3}(c) + \singlescore_{Q_3}(c') \\
    &\qquad + \singlescore_{Q_3'}(c') - \singlescore_{Q_3'}(c) + \singlescore_{Q_3'}(c_1) \\
    &= \singlescore_{Q_4}(c_1) + \singlescore_{Q_4}(c') - \singlescore_{Q_4}(c).
  \end{align*}
  We can rewrite the above inequality as:
  \begin{align*}
    \singlescore_{Q_4}(c) > \frac{1}{2}\Big( \singlescore_{Q_4}(c_1) +
    \singlescore_{Q_4}(c') \Big) \text{.}
  \end{align*}
  Since $\singlescore_{Q_4}(c_1) = \singlescore_{Q_4}(c')$ (election
  $Q_4$ is symmetric with respect to $c_1$ and $c'$) we get that
  $\singlescore_{Q_4}(c) > \singlescore_{Q_4}(c')$ and
  $\singlescore_{Q_4}(c) > \singlescore_{Q_4}(c_1)$. Thus, the
  $f_{m,1}$-score of $c$ in $Q_4$ is among the $f_{m,1}$-scores of the
  $(p-1)$ top-scored candidates. From our inductive assumption and
  from committee-enlargement monotonicity we infer that each winning
  committee in $Q_4$ must contain $c$. This contradicts the fact that
  $W'$ is winning in $Q_4$, and proves
  equation~\eqref{eq:equalDerivatives}. 
  Thus setting: $$\alpha =
  \frac{f_{m,p}(I^*) - f_{m,p}(J^*)}{\sum_{i \in I^*}f_{m,1}(i) -
    \sum_{i \in J^*}f_{m,1}(i)},$$ we get that for any two consecutive
  elements, $I$ and $J$, on path $\calS$ it holds that:
  \begin{align}
    f_{m,p}(I) - f_{m,p}(J) = \alpha\Big(\sum_{i \in I}f_{m,1}(i) -
    \sum_{i \in J}f_{m,1}(i)\Big) \text{.}
  \end{align}
  By a simple induction over the path $\calS$ we can show that the
  above equality holds for any $I$ and $J$ ($I$ and $J$ do not have to
  be consecutive elements in $\calS$). Consequently, we get that
  $f_{m,p}$ is a linear transformation of the function $g_{m,p} (I) =
  \sum_{i \in I}f_{m,1}(i)$, thus they yield the same committee
  scoring rule. This proves our inductive step, and completes the
  proof.
\end{proof}

Elkind et al.~\cite{elk-fal-sko-sli:j:multiwinner-properties} and
Barber\`a and~Coelho~\cite{bar-coe:j:non-controversial-k-names} point
out that committee-enlargement monotonicity is an extremely natural
requirement for multiwinner rules whose role is to select committees
of individually excellent candidates. For example, if some $k$
candidates are good enough to be shortlisted for receiving some award,
then increasing $k$ should not lead to any of them losing their
nominations.
Thus, intuitively, Theorem~\ref{thm:committee_monot_and_sep} says that
if one is interested in a committee scoring rule for choosing
individually excellent candidates, then one should look within the
class of separable rules. This refines and reinforces the
recommendation provided by Theorem~\ref{thm:weaklysep}, which
suggested looking among weakly separable rules.

Yet, one could challenge this recommendation. For example, SNTV is
separable, but it is also representation-focused and there is some
evidence that its behavior is closer to that of the
Chamberlin--Courant rule (which is seen as selecting committees
representing a diverse spectrum of opinions; recall the discussion
after Theorem~\ref{thm:representation_focused}), than to that of, say,
$k$-Borda (which is seen as selecting individually excellent
candidates). Such evidence is provided, for example, by Elkind et
al.~\cite{elk-fal-las-sko-sli-tal:c:2d-multiwinner}, who evaluated a
number of committee scoring rules experimentally, by computing their
results on elections obtained from several two-dimensional Euclidean
models and presenting them graphically.
Nonetheless, in real-life settings even SNTV is sometimes used for
choosing individually excellent candidates. As a piece of anecdotal
evidence, let us mention that while preparing this paper, we have ran
into a news article that listed $10$ best ski-jumpers of all time. The
criterion for inclusion on that list was the number of times a given
sportsman had won an individual competition of the ski-jumping World
Cup. In other words, all the individual World Cup competitions that
ever took place were seen as ``voters,'' ranking all the sportsmen
from the winner to the loser, and then SNTV was used to select the
``committee'' of $10$ best ski-jumpers of all
time.\footnote{Naturally, the sets of ski-jumpers participating in the
  contests were often different. Formally, we would say that the
  participating sportsmen were ranked according to their result in the
  competition and all the non-participating ones were ranked below, in
  some arbitrary order. Since we are using SNTV, the order in which
  the non-participants are ranked is irrelevant.}

Theorems~\ref{thm:weaklysep}~and~\ref{thm:committee_monot_and_sep}
have yet another interesting consequence. They imply that
committee-enlargement monotonicity of a committee scoring rule implies
its non-crossing monotonicity. This is somehow surprising, since the
two variants of monotonicity seem almost unrelated as one describes
how the result of an election changes if we increase the size of the
committee and the other one---what happens when we shift a member of a
winning committee in a preference relation of a voter.
\begin{corollary}
  If a committee scoring rule is committee-enlargement monotone, then it
  is also non-crossing monotone.
\end{corollary}

\section{Related Work}
\label{sec:related}

Over the last few years, multiwinner voting has attracted significant
interest within the computational social choice literature, but it has
also been studied for much longer within social choice theory and
within economics. Below we briefly review this literature (for a more
detailed review, we point the readers to the overview of Faliszewski
et al.~\cite{fal-sko-sli-tal:b:multiwinner-trends}; we have also
mentioned many related papers in the context of respective results).

Axiomatic studies of voting rules were initiated by
Arrow~\cite{arr:b:polsci:social-choice}, and in a somehow more narrow
framework, by May~\cite{mayAxiomatic1952}.  Single-winner scoring
rules are perhaps the best understood among single-winner election
systems.  Axiomatic characterizations of this class were provided,
e.g., by G{\"a}rdenfors~\cite{gardenfors73:scoring-rules},
Smith~\cite{smi:j:scoring-rules}, and
Young~\cite{you:j:scoring-functions}, and in a more general setting,
by Myerson~\cite{Myerson1995} and Pivato~\cite{pivato2013variable}.
More specific axiomatic characterizations of single-winner scoring
rules include those of the Borda rule~\cite{youngBorda, hansson76,
  fishburnBorda, smi:j:scoring-rules}, of the Plurality
rule~\cite{RichelsonPlurality, Ching1996298}, and of the Antiplurality
rule~\cite{Barbera198249} (see also the overviews of Chebotarev and
Shamis~\cite{che-sha:j:scoring-rules} and of
Merlin~\cite{merlinAxiomatic}).  Classic works on axiomatic properties
of multiwinner rules include those of Dummett~\cite{dum:b:voting},
Gehrlein~\cite{geh:j:condorcet-committee}, Felsenthal and
Maoz~\cite{fel-mao:j:norms}, Debord~\cite{deb:j:k-borda},
Ratliff~\cite{ratliff2003some}, and Barber\`a and
Coelho~\cite{bar-coe:j:non-controversial-k-names}.

Our work mostly builds on that of Elkind et
al.~\cite{elk-fal-sko-sli:j:multiwinner-properties} where the authors
introduced the class of committee scoring rules and many of the
notions on which we rely, such as candidate monotonicity, non-crossing
monotonicity and committee enlargement-monotonicity (regarding the
latter one, see also the work of Barber\`a and
Coelho~\cite{bar-coe:j:non-controversial-k-names}).  In particular,
Elkind et al.~\cite{elk-fal-sko-sli:j:multiwinner-properties}
identified the classes of (weakly) separable and
representation-focused rules and provided some of their basic
features.  OWA-based rules were introduced by
Skowron et al.~\cite{sko-fal-lan:c:collective}, who analyzed their
computational properties (and who, in fact, studied a somewhat more
general model).
Faliszewski et
al.~\cite{fal-sko-sli-tal:c:top-k-counting,fal-sko-sli-tal:c:paths}
introduced the class of top-$k$-counting rules and the $\ell_p$-Borda
and $q$-HarmonicBorda rules.

Recently Skowron et~al.~\cite{sko-fal-sli:t:axiomatic-committee}
characterized the class of committee scoring rules using the axioms of
consistency, symmetry, continuity, and weak efficiency.  Our paper can
be seen as complementary to theirs: They study committee scoring rules
as opposed to all the other multiwinner rules, whereas we focus on the
internal structure of the class.

Aziz et
al.~\cite{azi-bri-con-elk-fre-wal:j:justified-representation,azi-gas-gud-mac-mat-wal:c:approval-multiwinner}
studied a class of approval-based rules that is very similar to the
class of committee scoring rules (the class was first introduced by
Thiele~\cite{Thie95a} in the 19th century, but was forgotten for some
time; some of Thiele's rules were recalled by
Kilgour~\cite{kil-handbook} and then by Aziz et al.).  Lackner and
Skowron~\cite{lac-sko:t:approval-thiele} studied axiomatic
properties of these rules and highlighted their axiomatic similarity
to committee scoring rules.
Recently, monotonicity notions similar to those studied in this paper
were also considered in the context of approval-based multiwinner
rules~\cite{san-fis:t:approval-monotonicity,lac-sko:t:approval-thiele}.  For more general
discussions of the properties of approval-based rules we point the
reader to the work of Kilgour and
Marshall~\cite{kil-mar:j:minimax-approval}.

The study of computational properties of committee scoring rules in
general, and of specific rules, such as Chamberlin--Courant and
Proportional Approval Voting, has attracted significant
attention. This line of work has started with the paper of Procaccia
et al.~\cite{pro-ros-zoh:j:proportional-representation}, who have
shown that an approval-based variant of Chamberlin--Courant is
$\np$-hard to compute. The same result for the classic, Borda-based
variant was shown by Lu and
Boutilier~\cite{bou-lu:c:chamberlin-courant}. Betzler et
al.~\cite{bet-sli-uhl:j:mon-cc} considered parameterized complexity of
the rule, whereas the study of approximation algorithms was initiated
by Lu and Boutilier~\cite{bou-lu:c:chamberlin-courant}, who have given
a polynomial-time $(1-\frac{1}{e})$-approximation algorithm (this
algorithm, based on the greedy procedure of Nemhauser et
al.~\cite{nem-wol-fis:j:submodular} for submodular functions, has
since then been adapted to other committee scoring rules as well). Skowron et
al.~\cite{sko-fal-sli:j:multiwinner} improved this result by providing
a polynomial-time approximation scheme for the Borda-based variant;
Skowron and Faliszewski~\cite{sko-fal:j:maxcover} gave an FPT
approximation scheme for the approval-based variant (and argued why
the $(1-\frac{1}{e})$-approximation algorithm is the best we can hope
for among polynomial-time algorithms).
The complexity of Chamberlin--Courant was also studied in much depth
for various restricted domains, including the single-peaked
domain~\cite{bet-sli-uhl:j:mon-cc,cor-gal-spa:c:sp-width,pet:t:total-unimodularity},
the single-crossing domain~\cite{sko-yu-fal-elk:j:mwsc}, and a number
of
others~\cite{yu-cha-elk:c:multiwinner-trees,lac-pet:c:spoc,elk-pet:c:nice-trees}.
Faliszewski et
al.~\cite{fal-sli-sta-tal:c:cc-mon-clustering,fal-lac-pet-tal:c:csr-heuristics}
considered a number of heuristic algorithms.

Proportional Approval Voting (PAV) received a bit less attention than
the Chamberlin--Courant rule, but due to the work of Aziz et
al.~\cite{azi-bri-con-elk-fre-wal:j:justified-representation} on
justified representation, it is now being studied with increasing
interest (briefly put, Aziz et al\@. have shown that PAV is remarkably
good at providing committees that represent the voters proportionally,
as also confirmed by Brill et al.~\cite{bri-las-sko:c:apportionment};
see the work of Brill et al.~\cite{aaai/BrillFJL17-phragmen} for
another rule with similar properties). The rule was shown to be
$\np$-hard to
compute~\cite{azi-gas-gud-mac-mat-wal:c:approval-multiwinner,sko-fal-lan:c:collective},
but the standard greedy $(1-\frac{1}{e})$-approximation algorithm
works for it. Very recently, Byrka et
al.~\cite{byr-sko-sor:t:pav-approx} have shown a different, apparently
much more powerful algorithm.  The rule was also considered in the
context of restricted domains~\cite{pet:t:total-unimodularity}. FPT approximation schemes
for PAV and other OWA-based rules were provided by Skowron~\cite{sko:j:submodular}.

More general computational results regarding committee scoring rules
were provided by Skowron et al.~\cite{sko-fal-lan:c:collective}, who
studied the complexity and approximability of OWA-based rules, by
Faliszewski et al.~\cite{fal-sko-sli-tal:c:top-k-counting}, for
top-$k$-counting rules, by Peters~\cite{pet:t:total-unimodularity},
for OWA-based rules in the single-peaked domain, and by Faliszewski et
al.~\cite{fal-lac-pet-tal:c:csr-heuristics}, who introduced several
general-purpose heuristic algorithms.

Naturally, there exist many interesting multiwinner rules beyond the
class of committee scoring rules.  These include, for example, Single
Transferable Vote (see, e.g., the work of Tideman and
Richardson~\cite{tid-ric:j:stv}), a number of rules based on the
Condorcet
principle~\cite{AEFLS17:multiwinner-condorcet,bar-coe:j:non-controversial-k-names,dar:j:condorcet-hard,journals/scw/ElkindLS15,fis:j:condorcet-committee,fis:j:majority-committees,geh:j:condorcet-committee,ratliff2003some,sek-sik-xia:c:condorcet-bundling},
Monroe's rule~\cite{mon:j:monroe}, and different variants of the rule
invented by
Phragm\'en~\cite{Phra94a,Phra95a,Phra96a,Janson16arxiv,aaai/BrillFJL17-phragmen}.
For an overview of electoral systems used to select committees of
representatives in practice, we refer the reader to the book of
Lijphart and Grofman~\cite{grofmanChoosingElectoral}.

\section{Conclusion}
\label{sec:concl}

We have studied the class of committee scoring rules and explored the
interesting hierarchy formed by its subclasses studied to date
(including the class of decomposable rules introduced in this paper).
We have highlighted several fundamental properties of committee
scoring rules, ranging from the nonimposition property (i.e., that for
every committee and every nontrivial committee scoring rule, there is
an election where this committee wins uniquely under this rule), to
quite a varied landscape of monotonicity notions.  This allowed us to
partially match syntactic properties of such rules to their normative
properties.

There is a number of follow-up directions for this research. For
example, whole-committee monotonicity (where all the members of the
committee are shifted forward) is an interesting property. The
axiomatic characterization of OWA-based rules remains an open
problem. Further, it is interesting to see whether there exist other
properties, e.g., those which relate to proportionality of
representation, that can be used to characterize different (subclasses
of) committee scoring rules, or other multiwinner election systems
(the works of Aziz et
al.~\cite{azi-bri-con-elk-fre-wal:j:justified-representation} and
Lackner and Skowron~\cite{lac-sko:t:approval-thiele} made some headway
in this direction). A formal axiomatic study which would allow to
compare committee scoring rules to other multiwinner election systems
is an important, yet challenging question.

\subsubsection*{Acknowledgments.} 
Piotr Faliszewski was supported by the National Science Centre,
Poland, under project 2016/21/B/ST6/01509.  Arkadii Slinko was
supported by the Royal Society of NZ Marsden Fund UOA-254. Piotr
Skowron was supported by ERC-StG 639945 and by a Humboldt Research
Fellowship for Postdoctoral Researchers.  Nimrod Talmon was supported
by a postdoctoral fellowship from I-CORE ALGO.

\bibliographystyle{abbrv}
\bibliography{grypiotr2006}

\appendix

\section{Appendix}

In this section we prove two basic properties of committee scoring
rules.  The first one is nonimposition, which requires that for every
committee there is an election where this committee wins uniquely
(recall Definition~\ref{def:nonimposition} in
Section~\ref{sec:basic}).  We show that all non-degenerate committee
scoring rules have the nonimposition property.  In the proof we use
elections $\zeta(S)$ introduced in
Section~\ref{sec:com_mon_sep_rules}.

\newtheorem*{lemnon}{Lemma~\ref{lem:nonimposition}}

\begin{lemnon}
  Let $\calR_f$ be a committee scoring rule defined by a family of
  committee scoring functions $f=(f_{m,k})_{k\le m}$. The rule
  $\calR_f$ satisfies the nonimposition property if and only if every
  committee scoring function in $f$ is nontrivial.
\end{lemnon}
\begin{proof}
  The trivial rule does not satisfy the nonimposition
  property. If $\calR_f$ is nontrivial, then for each committee $W$,
  by Lemma~\ref{lemma:profile_zeta2}, election $\zeta(W)$ witnesses
  that $\calR_f$ satisfies nonimposition.
\end{proof}

The next observation will be useful a bit later.

\begin{observation}
  \label{obs:committeesW1W2}
  Consider two committees, $W_1$ and $W_2$, such that $|W_1 \cap W_2|
  = k-1$. In election $\zeta(W_1 \cup W_2) + \zeta(W_1 \cap W_2)$,
  committees $W_1$ and $W_2$ are the only winning ones. Indeed, by
  Lemma~\ref{lemma:profile_zeta2} we know that all $W$ with $W
  \subseteq W_1 \cup W_2$ are winning in $\zeta(W_1 \cup W_2)$ and
  that all $W$ with $W_1 \cap W_2 \subseteq W$ are winning in
  $\zeta(W_1 \cap W_2)$. The only two committees winning in both
  elections are $W_1$ and $W_2$. Since committee scoring rules satisfy
  consistency~\cite{sko-fal-sli:t:axiomatic-committee}, we conclude
  that $W_1$ and $W_2$ are the only winners in $\zeta(W_1 \cup W_2) +
  \zeta(W_1 \cap W_2)$.
\end{observation}

Next we give a proof of Lemma~\ref{lem:unique}, by showing that two
committee scoring functions (for a given number of candidates $m$ and
size $k$ of the committees) define the same rule (for these $m$ and~$k$)
if and only if they are linearly related.

\newtheorem*{ulem}{Lemma~\ref{lem:unique}}

\begin{ulem}
  \lemuni
\end{ulem}

\begin{proof}
  Let us fix $m$ and $k$.  Let $I_{\max} = (1, 2, \ldots, k)$ and
  $I_{\min} = (m-k+1, m-k+2, \ldots, m)$ be two committee positions,
  the former consisting of top $k$ positions and the latter consisting
  of bottom~$k$ ones.
  The statement of the lemma clearly holds when $f_{m,k}(I_{\max}) =
  f_{m,k}(I_{\min})$ as then $\calR_f$ is trivial, and so $g$ must be
  constant. Thus from now on we assume that $f(I_{\max}) >
  f(I_{\min})$. Let $h_{m,k}$ be a linear transformation of $g_{m,k}$
  such that $f_{m,k}(I_{\max}) = h_{m,k}(I_{\max})$ and
  $f_{m,k}(I_{\min}) = h_{m,k}(I_{\min})$. It is apparent that $h$ and
  $g$ implement the same multiwinner rule. We will show that $f_{m,k}
  = h_{m,k}$, which is sufficient to complete the proof. For the sake
  of contradiction let us assume that this is not the case.

  Since $f_{m,k} \neq h_{m,k}$, there must exist $I^*$ such that
  $f_{m,k}(I^*) \neq h_{m,k}(I^*)$; let us assume that $f_{m,k}(I^*) >
  h_{m,k}(I^*)$. There exists a sequence $\calS$ of committee
  positions from $[m]_k$, starting with $I_{\max}$, containing $I^*$,
  and ending in $I_{\min}$, such that for each two consecutive
  elements, $I$ and $J$, in the sequence (i.e., when $J$ appears right
  after $I$ in the sequence) it holds that:
  \begin{enumerate}[(i)]
  \item $|I \cap J| = k-1$, and
  \item $I$ dominates $J$.
  \end{enumerate}
  For instance, for $m = 5$, $k=2$ and $I^* = (2, 4)$, the sequence
  $\calS$ could be $((1, 2), (1, 3), (1, 4), (2, 4), (2, 5), (3, 5),
  (4, 5))$ (note that this sequence does not need to contain all
  possible committee positions and, thus, it is easy to form it).

  Consider function $\psi = f_{m,k} - h_{m,k}$. Since $f(I_{\max}) =
  h(I_{\max})$ and $f(I_{\min}) = h(I_{\min})$, we have that
  $\psi(I_{\min}) = 0$ and $\psi(I_{\max}) = 0$. Additionally, we know
  that $\psi(I^*) > 0$. Thus there exist committee positions $I, J,
  I', J' \in [m]_k$ such that $J$ is right after $I$ and $J'$ is right
  after $I'$ in the sequence $\calS$, and such that $\psi(I) \leq 0$,
  $\psi(J) > 0$, $\psi(I') > 0$, and $\psi(J') \leq 0$ (it might be
  the case that $J = I'$). That is: 
  \begin{align*}
    f(I) \leq h(I) ,\quad f(J)> h(J) ,\quad f(I') > h(I'), \quad
    \text{ and}\quad f(J') \leq h(J').
  \end{align*}
  Combining these inequalities, and taking into account that $I$
  dominates $J$, and that $I'$ dominates $J'$, we get that:
  \begin{align*}
    0 \leq f(I) - f(J) < h(I) - h(J) \quad \text{and} \quad f(I') -
    f(J') > h(I') - h(J') \geq 0\text{.}
  \end{align*}
  This means that there exist two positive integers, $x, y \in
  \naturals$, such that:
  \begin{align*}
       \frac{f(I) - f(J)}{f(I') - f(J')} < \frac{y}{x} <   \frac{h(I) - h(J)}{h(I') - h(J')}
  \end{align*}
  and, in consequence:
  \begin{align*}
    x(f(I) - f(J)) < y(f(I') - f(J')) \quad \text{and} \quad x(h(I) -
    h(J)) > y(h(I') - h(J')) \text{.}
  \end{align*}
  Let us fix two committees, $W_1$ and $W_2$, with $|W_1 \cap W_2| =
  k-1$, and consider an election $E$ with $x + y$ voters, where in $x$
  votes $W_1$ stands on position $I$ and $W_2$ on position $J$, and in
  $y$ votes $W_1$ stands on position $J'$ and $W_2$ stands on position
  $I'$. We can add to $E$ a sufficient number of copies of election
  $\zeta(W_1 \cup W_2) + \zeta(W_1 \cap W_2)$ (recall
  Section~\ref{sec:com_mon_sep_rules} for the definition of
  $\zeta$). By Observation~\ref{obs:committeesW1W2}, we know that in
  election $\zeta(W_1 \cup W_2) + \zeta(W_1 \cap W_2)$ only committees
  $W_1$ and $W_2$ are winning. Consequently, if we add a sufficient
  number of copies of this election to $E$, we can ensure that in $E$
  only $W_1$, $W_2$, or both $W_1$ and $W_2$ can be winners. Since the
  elections which we added to $E$ are symmetric with respect to $W_1$
  and $W_2$, the outcome of the election (i.e., whether $W_1$ or $W_2$
  is winning) depends only on election $E$. However, according to $f$
  committee $W_1$ has lower score than $W_2$, so the latter should be
  winning. Yet, by looking at $h$ we come to the opposite
  conclusion. This gives a contradiction and proves that $f = h$. This
  completes the proof.
\end{proof}

\end{document}